\documentclass[aps,twocolumn,groupedaddress,nofootinbib]{revtex4-1}

\usepackage{hyperref}
\usepackage{CJK}
\usepackage{amsmath}
\usepackage{graphicx}
\usepackage{dcolumn}
\usepackage{enumerate,amsthm,amssymb,color}
\usepackage{bbm}
\usepackage{bm}
\usepackage[french, english]{babel}
\usepackage{synttree}
\usepackage{caption}
\usepackage{subfig}
\usepackage{multirow}
\usepackage{xfrac}
\usepackage{fixltx2e}

\begin{document}

\title{State-injection schemes of quantum computation in Spekkens' toy theory}

\author{Lorenzo Catani}\email{lorenzo.catani.14@ucl.ac.uk}

\author{Dan E. Browne}\affiliation{University College London, Physics and Astronomy department, Gower St, London WC1E 6BT, UK}

\begin{abstract} 
Spekkens' toy theory is a non-contextual hidden variable model with an epistemic restriction, a constraint on what the observer can know about the reality. In reproducing many features of quantum mechanics in an essentially classical model, it clarified our understanding of what behaviour can be truly considered intrinsically quantum.
In this work, we show that Spekkens' theory can be also used to help understanding aspects of quantum computation - in particular an important subroutine in fault tolerant quantum computation called state injection. State injection promotes fault tolerant quantum circuits, which are usually limited to the classically efficiently simulatable stabilizer operations, to full universal quantum computation. We show that the limited set of fault tolerant operations used in standard state injection circuits can be realised within Spekkens' theory, and that state-injection leads to non-classicality in the form of contextuality. To achieve this, we extend prior work connecting Spekkens' theory and stabilizer quantum mechanics, showing that sub-theories of the latter can be represented within Spekkens' theory, in spite of the contextuality in qubit stabilizer quantum mechanics.  The work shines new light on the relationship between quantum computation and contextuality.
\end{abstract}

\maketitle

Spekkens' toy theory is a non-contextual hidden variable model made to advocate the epistemic view of quantum mechanics, where quantum states are seen as states of incomplete knowledge about a deeper underlying reality \cite{Spek1,Spek2}. The idea of the model is to reproduce quantum theory through a phase-space inspired theory with the addition of a constraint on what an observer can know about the ontic state (identified with a phase space point) describing the reality of a system. In the case of odd-dimensional systems the toy theory has been proven to be \emph{operationally equivalent} to qudit stabilizer quantum mechanics \cite{Ours}, where the latter is a subtheory of quantum mechanics restricted to eigenstates of tensors of Pauli operators, Clifford unitaries and Pauli measurement observables \cite{GottesPhD}. ``Operationally equivalent'' means that the two theories predict the same statistics of outcomes, given the same states, transformations and measurements. The equivalence between Spekkens' toy theory and qudit stabilizer quantum mechanics can be proven by representing qudit stabilizer quantum mechanics using Gross' Wigner functions \cite{Gross}. These turn out to be exactly equivalent to Spekkens' epistemic states and measurements. The measurement update rules are consistent and positiveness-preserving, while Clifford gates are mapped into consistent symplectic affine transformations \cite{Ours}.

In the case of qubits, the above equivalence does not hold. The toy theory is non-contextual by construction, whereas qubit stabilizer quantum mechanics shows state-independent contextuality, as witnessed by the Peres-Mermin square argument \cite{GHZ,Mermin,Peres}. This is reflected in the impossibility of finding a non-negative Wigner function that maps qubit stabilizer quantum mechanics into Spekkens' toy theory \cite{Wootters,Galvao}. Note that the non-negativity is needed in order to interpret the epistemic states and measurements as well-defined probabilities and thus state the operational equivalence.
Even if the toy theory and qubit stabilizer quantum mechanics are not operationally equivalent, some restricted versions of them do show the same statistics. Our aim is to identify the subtheories of Spekkens' toy theory that are operationally equivalent to subtheories of qubit stabilizer quantum mechanics. 
These will be closed subsets of operations, states and measurements in stabilizer quantum mechanics where states and measurements can be non-negatively represented by covariant Wigner functions.

In this work we use Spekkens' subtheories for an application in the field of quantum computation. More precisely, we use them to study state-injection schemes of quantum computation \cite{State-Injection}.
The latter are part of one of the currently leading models of fault tolerant universal quantum computation (UQC). This model is composed of a ``free'' part, which consists of quantum circuit that are efficiently simulatable by a classical computer (usually stabilizer circuits), and the injection of so-called ``magic'' resource states (that are usually distilled from many copies of noisy states through magic state distillation \cite{BravijKitaev}) that enable universal quantum computation.
In 2014  Howard et al. \cite{Howard} proved that in a state-injection scheme for qudits of odd prime dimensions, with the free part composed by stabilizer circuits, the contextuality possessed \emph{solely} by the magic resource is a necessary resource for universal quantum computation. 
In terms of systems of dimensions $2$ a similar result due to Delfosse et al. \cite{Delfosse} was derived for rebits, where the classical non-contextual free part is composed of Calderbank-Steane-Shor (CSS) circuits \cite{CSScodes}. However, as already pointed out, an analogue version of Howard's result for qubits cannot be found, since qubit stabilizer quantum mechanics is already contextual. Nevertheless it has been proven \cite{Vega} that in any state-injection scheme for qubits where we get rid of the state-independent contextuality, the contextuality possessed \emph{solely} by the magic resource is necessary for universal quantum computation. A more complete version of this result is also treated in \cite{Rauss2}, where a general framework for state-injection schemes of qubits with contextuality as a resource is provided. 

More precisely, in \cite{Rauss2}, Raussendorf et al. develop a framework for building non-negative and non-contextual subtheories of qubit stabilizer quantum mechanics via the choice of a phase function $\gamma(\lambda)$ on phase-space points, which defines the Weyl operators and consequently the Wigner functions (see equation \eqref{wigner}). They require that the allowed free measurements preserve the non-negativity of the Wigner function. In principle this requirement constrains the number of allowed observables. For this reason they also require tomographic completeness, \emph{i.e.} that any state can be fully measured by the observables allowed in the free part of the scheme, which guarantees that the set of free observables is large enough for the state-injection scheme to work. 
Unlike the case of measurements, they allow gates that introduce negativity in the Wigner functions, \emph{i.e.} non-covariant gates, the reason being that the gates can always be absorbed in the measurements without altering the outcome distribution of the computation. 
Furthermore, in \cite{WallmanBartlet}, Wallman and Bartlett address the issue of finding the subtheories of qubit quantum mechanics that are non-negative in certain quasiprobability representation (and so are classically simulatable and correspond to non-contextual ontological models). They construct the so-called $8-$state model, which can be seen as a generalisation of Spekkens' toy theory with an enlarged ontic space. The non-negativity for states and measurements is guaranteed by considering both the possible Wigner representations of a qubit \cite{Galvao}.  

It is important to point out that in the mentioned frameworks of \cite{Vega} and \cite{Rauss2}, the definition of state-injection schemes is broader than the one we consider here and it also includes schemes based on measurement-based quantum computation with cluster states \cite{MBQC}. We here consider only state-injection schemes as developed by Zhou et al. in \cite{State-Injection} -- defined in \ref{definition} --  like the ones in \cite{Howard} and \cite{Delfosse}, and we show that Spekkens' subtheories are an intuitive and effective tool to treat these cases.
We first use Spekkens' subtheories to represent the non-contextual free part of the known examples of state-injection schemes, both for qubits and qudits, where contextuality arises as a resource \cite{Howard,Delfosse}. These can be unified in the following framework (figure \ref{Ours}): $Spekkens' \; subtheory \; + \;Magic \; state(s)\;\rightarrow\; UQC$. 
Secondly, we prove in theorem \ref{Main} that qubit SQM can be obtained from a Spekkens' subtheory via state-injection since all its objects that do not belong to the Spekkens' subtheory, namely non-covariant Clifford gates, can be injected, where the circuit needed for the injection is always made of objects belonging to the Spekkens' subtheory. This means that Spekkens' subtheories contain all the tools for performing state-injection schemes of quantum computation. There is no need to consider bigger non-covariant subtheories in the free part. 

The proof of theorem \ref{Main} suggests a novel state-injection scheme, where contextuality is a resource, based on injection of $CCZ$ states. State-injection schemes with the related Toffoli (CCNOT) gates are already known \cite{Shor,Aharonov,Shi,Eastin,Jones1,Jones2,Paetznick}, but our scheme differs from them as our non-contextual free part of the computation -- a strict subset of the CSS rebit subtheory considered in \cite{Delfosse} -- is such that it is not possible to remove any object from it without denying the possibility of obtaining universal quantum computation via state-injection. The price to pay for this minimality is the injection of the control-Z state, $CZ\left |++\right\rangle,$ too (which also provides the Hadamard gate).
By analysing this example, we can associate different proofs of contextuality to the injection of different states, making a link between different classes of contextuality experiment and different states. More precisely, we show how the Clifford non-covariant $CZ$ gate (as well as the phase gate $S$) can provide proofs of the Peres-Mermin square argument and the GHZ paradox. Moreover, the injection of the $T\left |+\right\rangle$ magic state, where $T$ is the popular $\frac{\pi}{8}$ non-Clifford gate, allows, in addition to the previous proofs of contextuality (as $T^2=S$), also to obtain the maximum quantum violation of the CHSH inequality \cite{CHSH}.

\begin{figure}
\centering

{\includegraphics[width=.45\textwidth,height=.2\textheight]{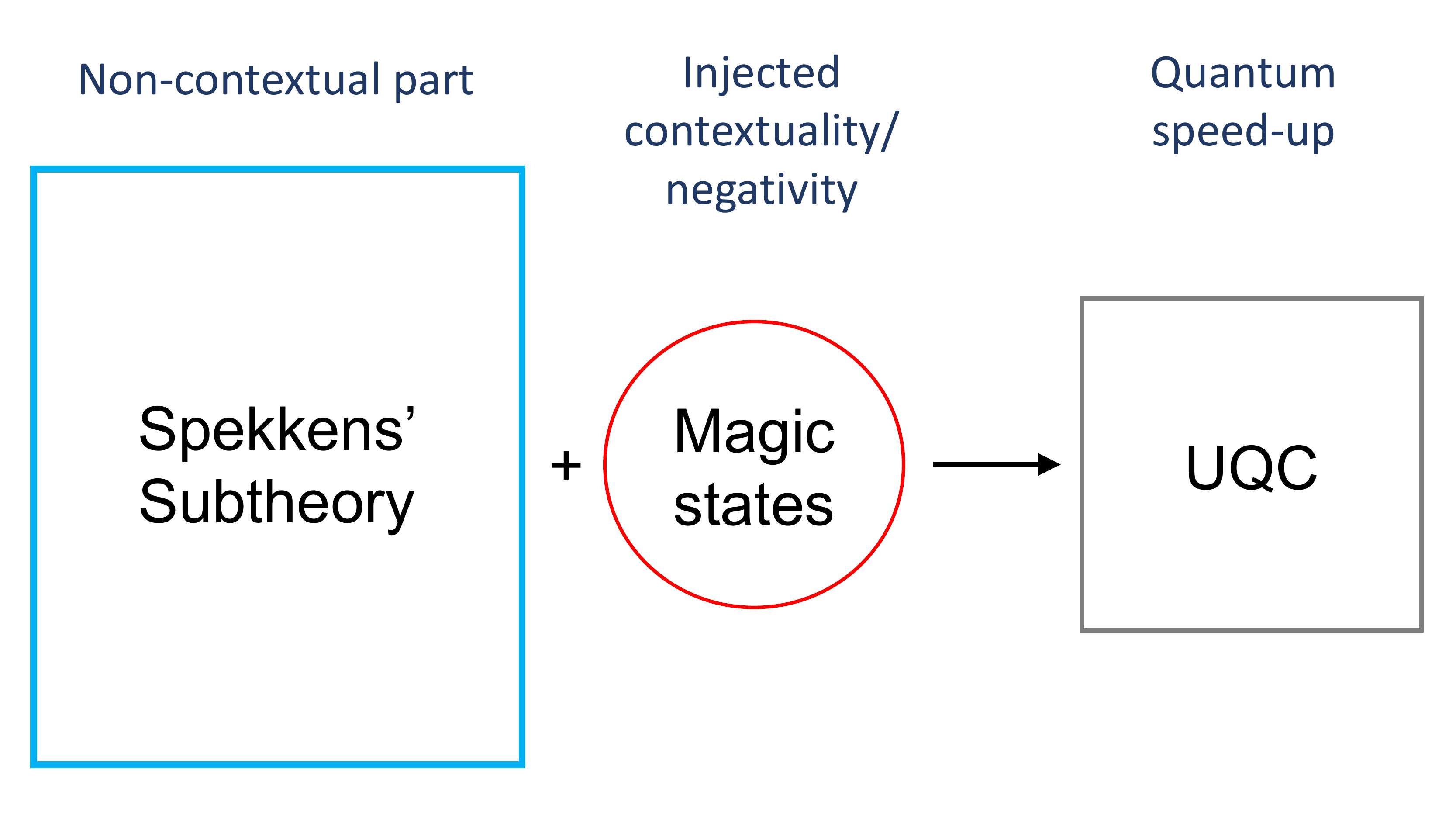}}
\caption[Computational scheme.]{\footnotesize{\textbf{Computational scheme.} Schematic representation of the computational scheme of the paper.}}
\label{Ours}
\end{figure}

In the remainder of the paper we start by covering some preliminary material on Spekkens' toy theory, Wigner functions and state-injection schemes in section \ref{zero}. We then provide the definition of a Spekkens' subtheory in section \ref{one}. In section \ref{two} we describe Howard's and Delfosse's cases for qudits \cite{Howard} and rebits \cite{Delfosse} respectively and we prove that they fit in our framework where the free parts of the computation, qudit stabilizer quantum mechanics and CSS rebits respectively, are Spekkens' subtheories. We then set the instructions to construct a Spekkens' subtheory from the choice of a non-negative Wigner function in section \ref{three}. We do so in line with \cite{Rauss2} and we find that the main difference from Raussendorf et al's formulation (and also from the 8-state model) consists of demanding for the covariance of the Wigner function with respect to the allowed gates. Moreover we do not demand for tomographic completeness. 
By exploiting this comparison we then prove theorem \ref{Main}. It basically shows that any state-injection scheme can be obtained from our framework, since any object not present in the considered Spekkens' subtheory can be injected by using an injection scheme made of objects in the Spekkens' subtheory. In section \ref{four} we provide a novel example of state-injection for qubits (rebits) based on CCZ magic states. We analyse the presence of different proofs of contextuality in correspondence of different state injected gates in section \ref{five} and we recap all the results and the future directions in the conclusion section.

\section{Preliminaries}
\label{zero}

In this work we focus on subtheories of quantum mechanics that we call Spekkens' subtheories, and their relation with state-injection schemes of quantum computation. 
We here recall some basic notions and definitions characterising Spekkens' toy theory, Wigner functions and state-injection schemes that will be useful for our purposes.

\subsection{Spekkens' toy theory.}
\label{zero_one}
Spekkens' toy theory is a non-contextual hidden variable theory with an \textit{epistemic restriction}, \emph{i.e.} a restriction on what can be known about the hidden variables (or \textit{ontic states}) describing the physical system \cite{Spek1,Spek2,Ours}. Ontic states are points $\lambda$ in  phase space, that we here assume to be discrete, $\Omega= \mathbb{Z}_d^n,$ where $n$ is the number of the $d$-dimensional subsystems composing the system. 
An Observable $\Sigma$ is defined as a linear functional in  phase space, \emph{i.e.} it takes the form $\Sigma=\sum_m(a_mX_m+b_mP_m),$ where $X_m,P_m$ denote fiducial variables, like position and momentum, that label the phase space, $a_m,b_m\in \mathbb{Z}_d$ and $m\in{0,\dots, n-1}.$ In the following we denote the ontic states and observables when considered in their vectorial representation, and vectors in general, with bold characters.
The outcome $\sigma$ of any observable measurement $\bold{\Sigma}=(a_0,b_0,\dots,a_{n-1},b_{n-1})$ given the ontic state $\boldsymbol{\lambda}=(x_0,p_0,\dots,x_{n-1},p_{n-1}),$ where $x_j,p_j,a_j,b_j\in \mathbb{Z}_d$ for every $j\in\{0,\dots,n-1\},$ is given by their inner product: $\sigma=\bold{\Sigma}^T\boldsymbol{\lambda}=\sum_j(a_jx_j+b_jp_j),$ where all the arithmetic is over $\mathbb{Z}_d.$

The epistemic restiction is called the \textit{classical complementarity principle} and it states that two observables can be jointly known only when their Poisson bracket is zero. 
This can be simply recast in terms of the symplectic inner product: $[\Sigma_1,\Sigma_2]\equiv\bold{\Sigma_1}^TJ\bold{\Sigma_2}=0,$ where $J= \bigoplus_{j=1}^{n} \begin{bmatrix} 0 & 1 \\ -1 & 0 \end{bmatrix}_j.$ 
A set of observables that can be jointly known by the observer represents a sub-space of $\Omega,$ known as an \emph{isotropic subspace} and denoted as $V=span\{\Sigma_1,\dots,\Sigma_n\}\subseteq\Omega,$ where $\Sigma_i$ indicates one of the generators of $V$.

The observer's best description of the physical system is a probability distribution $p(\lambda)$ over $\Omega$ which is called the \emph{epistemic state}. 
It is defined as \begin{equation}\label{epistemic} P_{(V,\bold{w})}(\lambda)=\frac{1}{N}\delta_{V^{\perp}+\bold{w}}(\lambda),\end{equation} where the perpendicular complement of $V$ is, by definition, $V^{\perp}=\{a\in\Omega \; |\; \bold{a}^T \bold{b}=0\; \forall\; b\in V \}.$  The subspace of known variables $V$ specifies the observables $\Sigma_j$ that are known and the evaluation shift vector $\bold{w}\in V$ the values $\sigma_j$ that they take, $\bold{\Sigma}_j^T\cdot \bold{w}=\sigma_j.$ 
The indicator function $\delta_{V^{\perp}+\bold{w}}(\lambda)$ takes value $1$ when $\lambda$ belongs to the set $V^{\perp}+\bold{w}$ and $0$ otherwise; $N$ is a normalisation factor. 

The allowed transformations in Spekkens' theory are the ones that preserve the epistemic restriction, \emph{i.e.} the symplectic affine transformations $G$ in  phase space (in general a subset of the permutations), namely $G(\lambda)=S\boldsymbol{\lambda} + \bold{a},$ where $S$ is a symplectic matrix and $\bold{a}\in\Omega$ a translation vector. The elements $\Pi_k$ of a sharp measurement $\Pi,$ where the integer $k$ denotes the outcome associated with the measurement, have an epistemic representation analogue to the states of the form of equation \eqref{epistemic}. Here, we are in a dual representation between states and observables, where we denote with $V_{\Pi}$ and $\bold{r}_k$ the subspace of known observables and the evaluation shift vector associated with the $k$-th element of the sharp measurement, respectively. Figure \ref{SpekSubToffoli} provides a graphic representation of epistemic states, observables and evolutions for some particular examples of quantum states, observables and gates. 

\begin{figure}
\centering
\subfloat[][\label{a}]
{\includegraphics[width=.48\textwidth,height=.19\textheight]{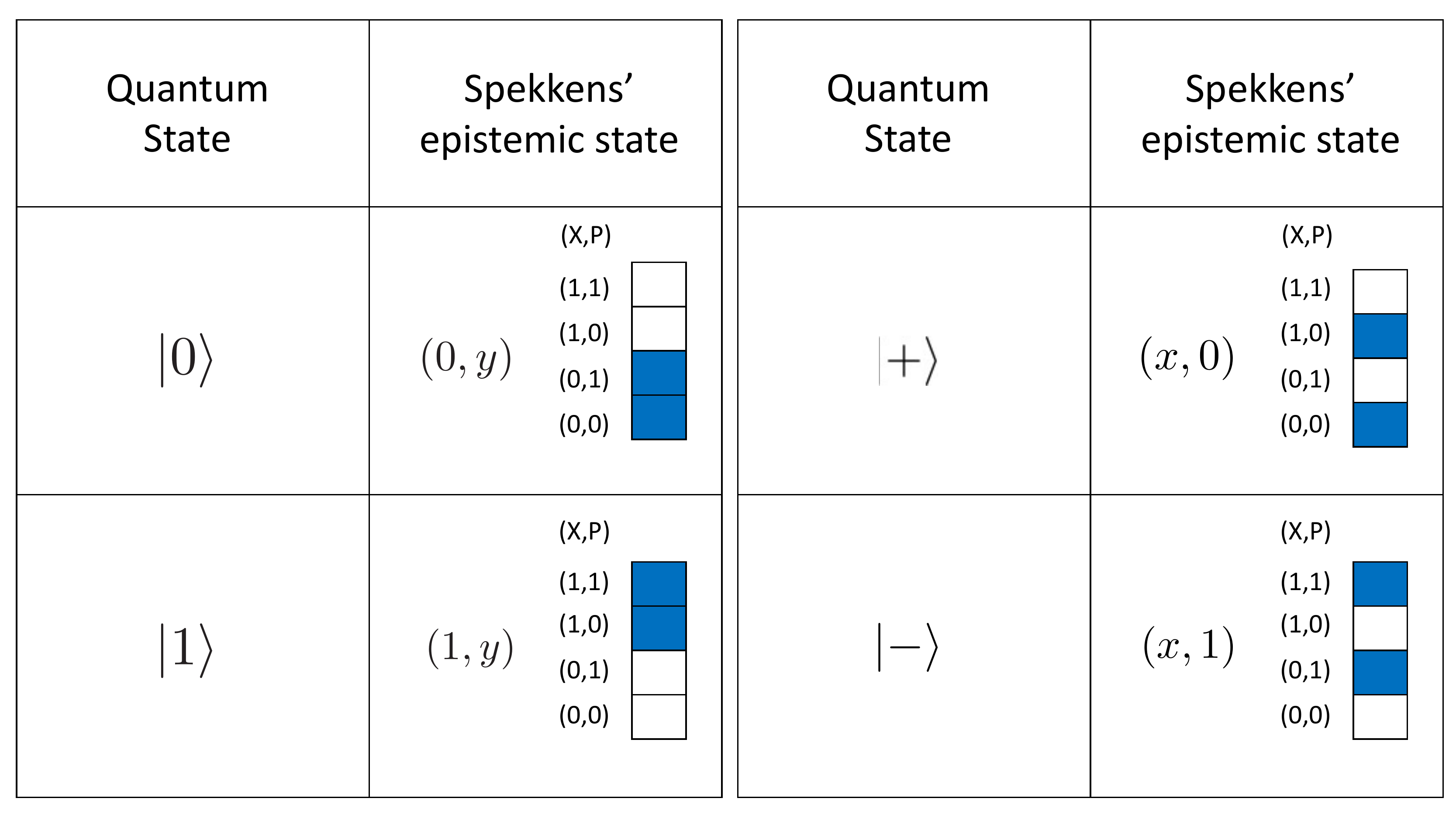}}\hfil
\subfloat[][\label{b}]
{\includegraphics[width=.4\textwidth,height=.165\textheight]{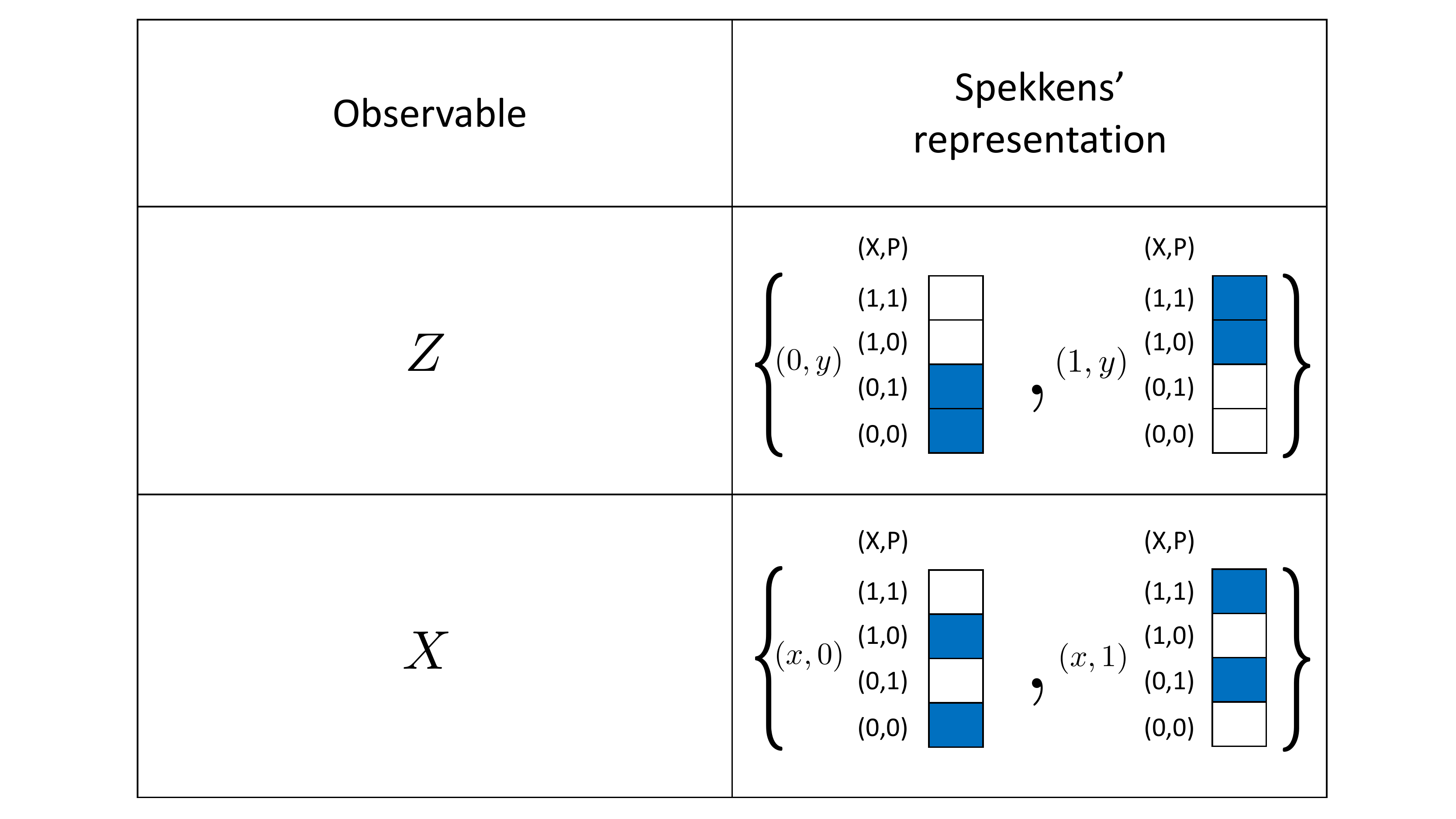}}\hfil
\subfloat[][\label{c}]
{\includegraphics[width=.48\textwidth,height=.19\textheight]{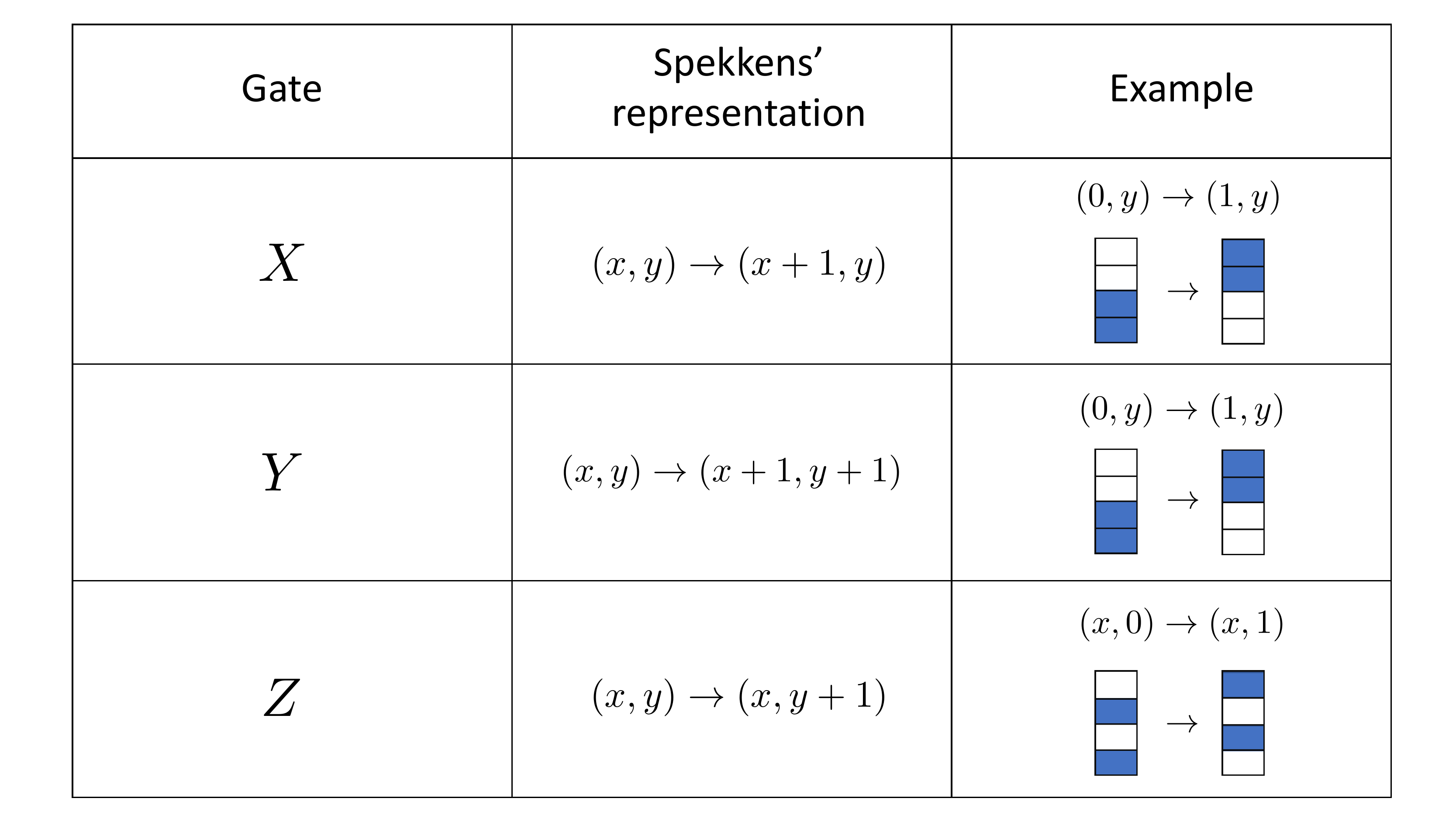}}\hfil
\subfloat[][\label{d}]
{\includegraphics[width=.48\textwidth,height=.19\textheight]{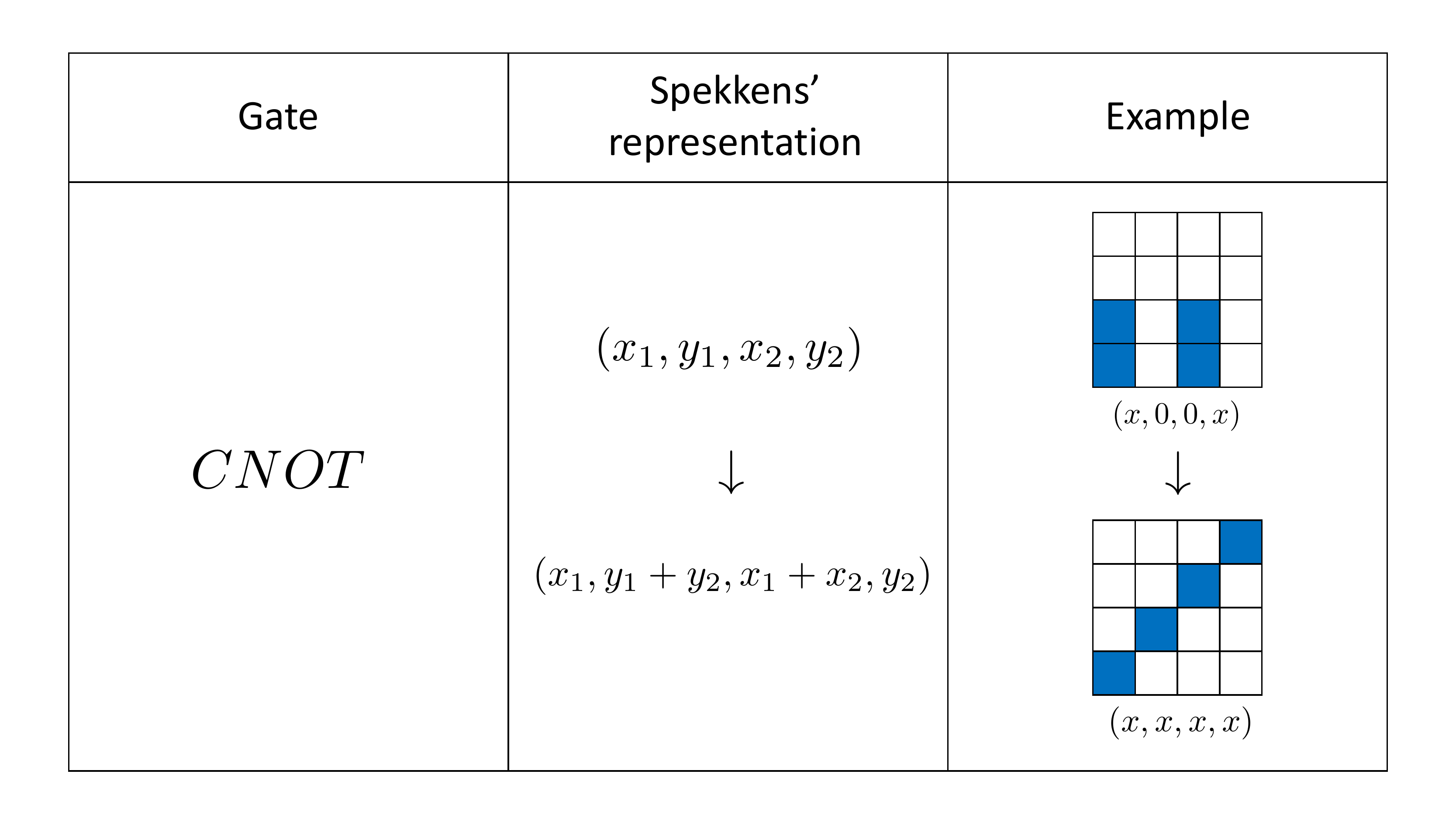}}
\caption[Representation of a non-contextual subtheory of qubit stabilizer quantum mechanics in Spekkens' toy model.]{\footnotesize{\textbf{Representation of a non-contextual subtheory of qubit stabilizer quantum mechanics in Spekkens' toy model.} In the figure above the allowed pure states \ref{a}, observables \ref{b} and gates \ref{c},\ref{d} of the non-contextual subtheory of qubit stabilizer quantum mechanics considered in the proof of theorem \ref{Main} are represented in Spekkens' toy model, according to the definitions of epistemic states, observables and evolutions (gates) defined in subsection \ref{zero_one}.  We have also indicated the probability distributions associated to the epistemic states, where $x,y\in\mathbb{Z}_d,$ and how the gates act on them. In the examples of figures \ref{c},\ref{d} we have considered the scenarios corresponding to acting with $X,Y$ on $\left |0\right\rangle,$ with $Z$ on $\left |+\right\rangle$ and with $CNOT$ on $\left |+0\right\rangle$ (thus obtaining the Bell state $\frac{\left |00\right\rangle +\left |11\right\rangle}{\sqrt{2}}).$ More detailed examples can be found in \cite{Spek2} and \cite{Ours}.}}
\label{SpekSubToffoli}
\end{figure}

\subsection{Wigner functions.}
\label{zero_two}
Wigner functions are a way of recasting quantum mechanics in the phase space framework\cite{Wigner,Wootters,Galvao}. They are called quasi-probability distributions,  not proper probability distributions, because they can take also negative values. Nevertheless their marginals represent probability distributions of measurement outcomes. The feature of negativity is usually associated with a signature of non-classicality \cite{Veitch,Spek4}.
We here define Wigner functions following \cite{Rauss2}.
The Wigner function of a quantum state $\rho$ 
is defined by the function $\gamma,$
\begin{equation}\label{wigner}W^{\gamma}_{\rho}(\lambda)=Tr(A^{\gamma}(\lambda)\rho),\end{equation} where the phase-point operator is \begin{equation}\label{phasepoint} A^{\gamma}(\lambda)= \frac{1}{N_{\Omega}}\sum_{\lambda' \in\Omega}\chi([\lambda,\lambda'])T^{\gamma}(\lambda'),\end{equation} and the Weyl operator is defined by \begin{equation}\label{weyl} T^{\gamma}(\lambda)=w^{\gamma(\lambda)}Z(\lambda_Z)X(\lambda_X),\end{equation} where the phase-space point is $\lambda=(\lambda_Z,\lambda_X)\in \Omega.$ We will omit the superscript $\gamma$ in the future in order to soften the notation. The normalisation $N_{\Omega}$ is such that $Tr(A(\lambda))=1.$
The functions $\chi$ and $w$ will be appropriately characterised in the case of qubits and qudits in section \ref{two}, as well as $\gamma.$ 
The operators $Z(\lambda_Z),X(\lambda_X)$ represent the (generalised) Pauli operators, \begin{equation}X(\lambda_X)=\sum_{\lambda'_X\in\mathbb{Z}_d}\left | \lambda'_X-\lambda_X \right\rangle \left\langle \lambda'_X \right |,\end{equation}\begin{equation}Z(\lambda_Z)=\sum_{\lambda_X\in\mathbb{Z}_d}\chi(\lambda_X\cdot\lambda_Z)\left | \lambda_X \right\rangle \left\langle \lambda_X \right |.\end{equation} 
The most important property of the Wigner functions for this work is the property of covariance, which means that (in accordance with the definition used by David Gross in \cite[theorem~7]{Gross}), for all  allowed states $\rho$ in the theory,  \begin{equation}\label{covariance}W_{U\rho U^{\dagger}}(\lambda)=W_{\rho}(S\boldsymbol{\lambda}+\bold{a}),\end{equation} where $S$ is a symplectic transformation and $\bold{a}$ is a translation vector. This property guarantees that the transformations in quantum mechanics correspond to symplectic affine transformations in phase space.

\subsection{State-injection schemes of quantum computation.}
\label{zero_three}
In this work we consider state-injection schemes of quantum computation as developed in \cite{State-Injection}. They represent one of the leading models for fault tolerant universal quantum computation when combined with the magic state distillation procedure due to Bravyi and Kitaev in \cite{BravijKitaev}. The latter allows to distill non-stabilizer magic states from noisy copies of quantum states with a high threshold for the error rate (about $14.6\%$), given a setting where only Clifford gates are fault tolerant. 
The key idea of state-injection is that a non-Clifford gate can be implemented with the combination of the magic state, a Clifford group circuit and Pauli measurements. More precisely, we can define state-injection schemes for implementing any diagonal unitary gate $U$ as follows.

\newtheorem{defn}{Definition}
\label{definition}
\begin{defn}[Zhou-Leung-Chuang state-injection \cite{State-Injection}]
A  \textbf{state injection} of an $n$-qubit unitary gate $U$ is a quantum circuit implementing $U$ composed of the following elements (figure \ref{QCSI}). 
\begin{itemize}
\item The injected state $U| + \rangle^{\otimes n}.$
\item $n$ CNOT gates, applied transversally.
\item $n$ Pauli $Z$ measurements, with the output of the $j$th measurement denoted as $s_j=(-1)^{m_j},$ where $m_j\in\{0,1\}.$
\item The correction gate $UX^mU^{\dagger},$ where $m=m_1\ldots m_n$ is the bitstring of meausurement outcomes and $X^m=X^{m_1}\otimes\cdots\otimes X^{m_n}.$ 
\end{itemize}

\end{defn}

As was proved in \cite{State-Injection}, any diagonal gate $U$ can be implemented via a state injection circuit of this form. However, for the injection of $U$ to succeed deterministically, the unitary correction $UX^mU^{\dagger}$ must be implementable in the model.
Figure \ref{QCSI} depicts the state-injection scheme just defined. 

\begin{figure}
\centering
{\includegraphics[width=.45\textwidth,height=.2\textheight]{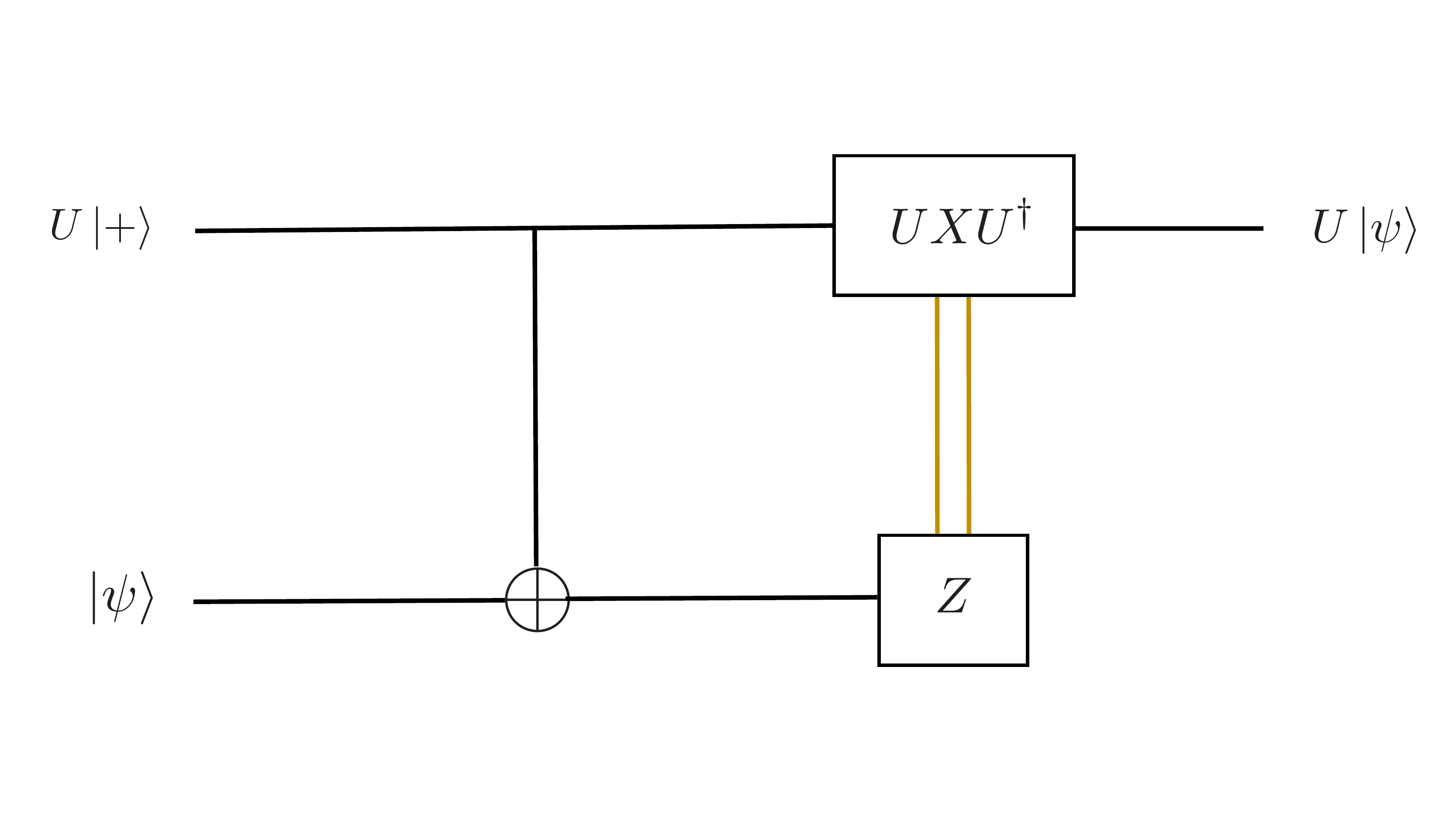}}
\caption[state-injection schemes of quantum computation.]{\footnotesize{\textbf{state-injection schemes of quantum computation.} The state-injection schemes that we consider in this work are the ones developed by Zhou, Leung and Chuang in \cite{State-Injection}. The diagonal gate $U$ can be injected in the circuit by using objects that are allowed in the free part of the computation. The injected state is $U\left | + \right \rangle,$ which is subjected to a controlled not with the input state $\left | \psi \right \rangle.$ Conditioned on the outcome of the measurement of the Pauli $Z$ on the state $\left | \psi \right \rangle$ after the $CNOT,$ the correction $UXU^{\dagger}$ is applied to the state $U\left | + \right \rangle.$ At the end we obtain the gate $U$ applied to the input state $\left | \psi \right \rangle.$}}
\label{QCSI}
\end{figure}

\section{Characterisation of Spekkens' subtheories}
\label{one}

A Spekkens' subtheory is defined as a set of quantum states, transformations and measurements, $(\mathcal{S},\mathcal{T},\mathcal{M}),$ which satisfies the following conditions.

\begin{enumerate}

\item \label{Sub} \emph{Subtheory.} The set must be \emph{closed}, which means that any allowed gate cannot bring from one allowed state to a non-allowed one. 
\begin{equation} \label{gates}\forall \; U\in \mathcal{T}, U\rho U^{\dagger}\in \mathcal{S} \; \forall \rho\in \mathcal{S}.\end{equation} 

\item \label{SR} \emph{Spekkens representability.} There must be an operational equivalence between the subtheory of quantum mechanics $(\mathcal{S},\mathcal{T},\mathcal{M}),$ defined by sets of quantum states, transformations and measurements, and a subtheory of Spekkens' toy theory $(\mathcal{S}_s,\mathcal{T}_s,\mathcal{M}_s),$ defined by sets of epistemic states, symplectic affine transformations and measurements as defined in the previous section. 

The operational equivalence means that the statistics of the two subtheories $(\mathcal{S},\mathcal{T},\mathcal{M})$ and $(\mathcal{S}_s,\mathcal{T}_s,\mathcal{M}_s)$ are the same. We state this equivalence by finding a \emph{non-negative} Wigner function that maps the states $\rho$ and the measurement elements $\Pi_k$ in $(\mathcal{S},\mathcal{M})$ to epistemic states and measurement elements in $(\mathcal{S}_s, \mathcal{M}_s),$ \emph{i.e.}
\begin{equation} \label{state} W_{\rho}(\lambda)=\frac{1}{N}Tr(\rho A(\lambda))=\frac{1}{N}\delta_{(V^{\perp}+\bold{w})}(\lambda);\end{equation}
\begin{equation} \label{meas} W_{\Pi_k}(\lambda)=\frac{1}{N'}Tr(\Pi_k A(\lambda))=\frac{1}{N'}\delta_{(V_{\Pi}^{\perp}+\bold{r}_k)}(\lambda).\end{equation}
The $N,N'$ are the normalisation factors so that $\sum_{\lambda\in\Omega}W(\lambda)=1$ for all the above Wigner functions. 
Notice at this point that the measurement update rules in \cite{Ours} guarantee that also a state after a measurement is non-negatively represented if the measurement and the original state have non-negative Wigner functions, as they involve only sums and products of Wigner functions. Moreover we are considering subtheories with duality between states and measurement elements, therefore it is enough to check only the properties of the Wigner functions of states (or measurements) to guarantee Spekkens representability.

The operational equivalence in terms of transformations is implied if the Wigner function \eqref{state} satisfies the property of covariance \eqref{covariance} for the allowed unitaries $U\in\mathcal{T}$, which guarantees that the transformations in quantum mechanics correspond to transformations that preserve the epistemic restriction in Spekkens' theory.
Notice that the property of covariance is defined in terms of Wigner functions and not directly in terms of the phase point operators. \footnote{This will make a difference only in section \ref{CSSsection} where we consider factorisable Wigner functions for the CSS rebit case.} Therefore we do not necessarily need to demand for the standard Wigner function of the transformation, defined as $W_{U}(\lambda|\lambda')=\frac{1}{N'}Tr(A(\lambda)UA(\lambda')U^{\dagger}),$ to be non-negative in all its elements, once the previous requirements, non-negativity and covariance of $W_{\rho}$, are satisfied. The transition matrix corresponding to the allowed permutation of the phase points can be always found, as shown by the following lemma.

\newtheorem{WignerFunctionTransformation}{Lemma}
\begin{WignerFunctionTransformation}\label{WignerFunctionTransformation}
Given non-negative Wigner function representations, $W_{\rho},W_{\rho'},$ of any two allowed states $\rho,\rho' \in \mathcal{S}$ such that $\rho'=U\rho U^{\dagger},$ where $U\in \mathcal{T},$ and covariance holds, \emph{i.e.} $W_{\rho'}(\lambda)=W_{\rho}(S\boldsymbol{\lambda}+\bold{a}),$ there always exists a (non-negative) transition matrix $P_U: \Omega \times \Omega \rightarrow [0,1]$ representing the transformation $U\in  \mathcal{T},$ \begin{equation} \label{transf} P_{U}(\lambda|\lambda')=\frac{1}{N''}\delta_{\lambda,S\boldsymbol{\lambda}'+\bold{a}},\end{equation} where $N''$ is the normalisation factor, such that  \begin{equation}\label{transproof}W_{\rho'}(\lambda)=\sum_{\lambda'\in\Omega}P_U(\lambda|\lambda')W_{\rho}(\lambda').\end{equation}
\end{WignerFunctionTransformation}

\begin{proof}
A matrix made of non-negative elements $P_U(\lambda|\lambda')$ proportional to Kronecker deltas always exists because it corresponds to the transition matrix representing the permutation that brings $W_{\rho}$ to $W_{\rho'}$. More precisely, non-negative solutions $P_U(\lambda |\lambda')$ to the equations \eqref{transproof} for every $\lambda$, given the non-negative $W_{\rho}(\lambda'),W_{\rho'}(\lambda)$ defined in \eqref{wigner}, always exist. For every fixed $\lambda,$ the  $P_U(\lambda |\lambda')$ are vectors with all zero components apart from one, \emph{i.e.} they are proportional to Kronecker deltas. 
The covariance property \eqref{covariance} guarantees that this permutation corresponds to a symplectic affine transformation on the phase space points (independent on the state $\rho$ that $U$ is acting on). 
\end{proof}

The non-negative functions \eqref{state}, \eqref{meas} and \eqref{transf} can be interpreted as probability distributions and guarantee that the theories $(\mathcal{S},\mathcal{T},\mathcal{M})$ and  $(\mathcal{S}_s,\mathcal{T}_s,\mathcal{M}_s)$ are operationally equivalent, \emph{i.e.} they provide the same statistics: \begin{equation}\begin{split} p(k) & = Tr(\Pi_k U\rho U^{\dagger})\\ & =\sum_{\lambda\in\Omega}W_{\Pi_k}(\lambda)\sum_{\lambda'\in\Omega}P_{U}(\lambda|\lambda')W_{\rho}(\lambda').\end{split}\end{equation}

\end{enumerate}

To sum up, a Spekkens' subtheory is a (closed) subtheory of quantum mechanics whose states (and measurements) are represented by non-negative and covariant Wigner functions. 
We say that a Spekkens' subtheory is \emph{maximal} if the set $(\mathcal{S},\mathcal{T},\mathcal{M})$ is such that by adding either another state, gate or observable to the set of allowed states, transformations and observables contradicts at least one of the conditions above, \emph{i.e.} it is no longer a subtheory or Spekkens representable. We will also talk about \emph{minimal} non-contextual subtheories of stabilizer quantum mechanics meaning those subtheories that can no longer be used for state-injection schemes after the removal of just one object from them.

\section{Wigner functions for Spekkens' subtheories}
\label{two}

We now identify the functions $\gamma$ defining the Wigner functions in equation \eqref{wigner} that allow us to show that the known examples of state-injection schemes with contextuality as a resource, \cite{Howard} and \cite{Delfosse}, fit into the framework depicted in figure \ref{Ours}, \emph{i.e.} that the free parts of those schemes are Spekkens' subtheories.

\subsection{Qudit case}

In the case of qudits of odd dimensions Gross' theorem \cite{Gross} guarantees that there is a non-negative Wigner representation of all pure stabilizer states. This Wigner function is covariant and also Clifford transformations and Pauli measurements are non-negatively represented. Thus Gross' Wigner function proves the operational equivalence, in odd dimensions, between stabilizer quantum mechanics and the whole Spekkens' toy theory, as shown in \cite{Spek2} and \cite{Ours}. 

Gross' Wigner function for odd dimensional systems (qudits) is defined according to equation \eqref{wigner}, where $\chi(a)=e^{\frac{2\pi i}{d}a},$ for $a\in \mathbb{Z}_d,$ and $w^{\gamma(\lambda)}=\chi(-2^{-1}\gamma(\lambda)).$ The function $\gamma$ is given by $\gamma(\lambda)=\lambda_X\cdot \lambda_Z,$ where the ``$\cdot$'' denotes the inner product. 

In the scheme of Howard et al. \cite{Howard} where they prove that contextuality is a resource for universal quantum computation, the free part of the computation is given by stabilizer quantum mechanics in odd prime dimensions, which, by Gross' Wigner functions, is a maximal Spekkens' subtheory. \footnote{Stabilizer quantum mechanics in odd dimensions is the unique maximal Spekkens' subtheory, since it coincides with the whole Spekkens' theory.}

\subsection{Qubit case}
\label{CSSsection}

In the case of qubits an analogue of Gross' Wigner function for stabilizer quantum mechanics does not exist \cite{Wootters,Galvao}, as some negative stabilizer states are present for any possible choice of Wigner functions. This is related to the contextuality shown by qubit stabilizer quantum mechanics (see, for example, the Peres-Mermin square \cite{Peres}). 
Nevertheless it is possible to state a similar result to Howard et al.'s by restricting the free part of the computation to a strict non-contextual and positive subtheory of qubit stabilizer quantum mechanics, as shown by the result of Delfosse et al. in 2015 \cite{Delfosse}.

We start by describing the set of allowed states/gates/observables $(\mathcal{S}_r,\mathcal{T}_r,\mathcal{M}_r)$  considered by Delfosse et al. 
The set $\mathcal{S}_r,$ a subset of the stabilizer states, is composed by Calderbank-Steane-Shor (CSS) states \cite{CSScodes}, \emph{i.e.} stabilizer states $\left |\psi \right\rangle,$ whose corresponding stabilizer group $S(\left |\psi \right\rangle)$ decomposes into an $X$ and a $Z$ part; \emph{i.e.} $S(\left |\psi \right\rangle)=S_X(\left |\psi \right\rangle)\times S_Z(\left |\psi \right\rangle),$ where all elements of $S_X(\left |\psi \right\rangle)$ and $S_Z(\left |\psi \right\rangle)$ are of the form $X(\bold{q})$ and $Z(\bold{p}),$ respectively, where $\bold{q},\bold{p}\in \mathbb{Z}_2^n$. 
CSS states are the eigenstates of the allowed observables belonging to the set $M_r,$ \begin{equation}\label{CSSobservables} \mathcal{M}_r = \{X(\bold{q}),Z(\bold{p})|\bold{q},\bold{p} \in \mathbb{Z}_2^n\}.\end{equation}
The set of allowed transformations is composed by the CSS preserving gates, subset of the Clifford group $C_n$ (which is the group of unitaries that maps Pauli operators to Pauli operators by conjugation), \begin{equation}\begin{split} \label{CSSgates} \mathcal{T}_r & =\{g\in C_n | g\left |\psi \right\rangle \in \mathcal{S}_r, \forall \left |\psi \right\rangle \in \mathcal{S}_r\}\\ &=\left\langle \bigotimes_{i=1}^n H_i,CNOT(i,j),X_i,Z_i\right\rangle,\end{split}\end{equation} where $i,j\in\{1,2,\dots,n\}$ and $i\neq j.$
The universal quantum computation is reached by injecting two particular magic states to the free subtheory of CSS rebits just described \cite{Delfosse}.
 
The Wigner function used by Delfosse et al. to prove their result is given by \begin{equation} \label{DelfosseWf} A_r(\lambda)=\frac{1}{2^n}\sum_{T(\lambda')\in \mathcal{A}}(-1)^{\left\langle \lambda,\lambda' \right\rangle}T(\lambda'),\end{equation}
where $T(\lambda)=Z(\bold{p})X(\bold{q}),$ $\lambda=(\bold{q},\bold{p}),$ and $\mathcal{A}=\{T(\lambda)|\bold{q}\cdot\bold{p}=0 \; \text{mod}\; 2\},$ where $\bold{q}\cdot\bold{p}$ denotes the inner product. The set $\mathcal{A}$ is the set of inferred observables. ``Inferred'' means that these observables may not be directly measurable, but they can be inferred by multiple measurements. For example in the case of two qubits, the set $M_r$ and $\mathcal{A}$ are $M_r = \{\mathbb{I}\mathbb{I},\mathbb{I}X,\mathbb{I}Z,X\mathbb{I},Z\mathbb{I},XX,ZZ\},$ and $\mathcal{A}=\{\mathbb{I}\mathbb{I},\mathbb{I}X,\mathbb{I}Z,X\mathbb{I},Z\mathbb{I},XX,ZZ,XZ,ZX,YY\},$ \emph{i.e.} the set of all rebits observables. Notice that we removed the ``$\otimes$'' symbol for the tensor product in order to simplify the notation.
This Wigner function is always non-negative for CSS states (\cite{Delfosse}) and it is covariant. In terms of the definition provided in equation \eqref{wigner}, the function $\gamma$ is $\gamma(\lambda)=0,$   
$\chi(a)=(-1)^a$ and $w^{\gamma(\lambda)}=1.$ This choice guarantees that the phase point operators are Hermitian. However the price to pay for the Hermitianity in this case is the non-factorisability of the Wigner function, \emph{i.e.} the Wigner function is composed by phase-point operators of $n$ qubits that are not given by the tensor products of the ones for the single qubit, \emph{e.g.} $A_r((0,0),(0,0)) \neq A_r(0,0)\otimes A_r(0,0).$

Before we proceed, one may wonder whether the non-factorisability of the Wigner function is necessary to treat the CSS case and preserve the non-negativity and covariance. Here we show that it is not. We define a Wigner function that, we argue, is more in line 
with the construction of Spekkens' model, where the ontic space of $n$ systems is made by the cartesian products of individual systems' subspaces.
The non-negative, covariant and factorisable Wigner function for the CSS theory is built out from the single-qubit phase-point operators \begin{equation}\label{factorisablePPO} A_f(0,0)=\mathbb{I} + X+Z+iY.\end{equation} The phase point operators $ A_f(0,1), A_f(1,0), A_f(1,1)$ are given by applying the Pauli $X,Y,Z$ respectively by conjugation on $A_f(0,0).$ The phase point operators of many qubits are given by tensor products of the ones for single qubits $ A_f(0,0), A_f(0,1), A_f(1,0), A_f(1,1).$ Notice that the phase point operators are not Hermitian; however the allowed observables are only present in their Hermitian part (\emph{e.g.} for the single qubit in $\mathbb{I} + X+Z$).
We now need to prove the following lemma.

\newtheorem{CSSWignerFunctions}[WignerFunctionTransformation]{Lemma}
\begin{CSSWignerFunctions}\label{CSSWignerFunctions}
The Wigner function of Delfosse et al. $W_r(\lambda)=Tr(\rho A_r(\lambda)),$ given by \eqref{DelfosseWf},  is equivalent to the factorisable Wigner function $W_f(\lambda)=Tr(\rho A_f(\lambda)),$ given by \eqref{factorisablePPO}, for any $\rho\in\mathcal{S}_r.$
\end{CSSWignerFunctions}

\begin{proof}
What we need to prove is actually that $A_r(\lambda)=\mathcal{H}(A_f(\lambda)),$ where $\mathcal{H}(A_f(\lambda))$ indicates the Hermitian part of the phase point operator $A_f(\lambda).$ The non-Hermitian part of $A_f(\lambda)$ has zero contribution to the Wigner function. It is always composed of tensors of mixtures of Pauli operators with an odd number of $Y$'s, that never form allowed observables and so are never in the stabilizer group of any $\rho\in\mathcal{S}_r$. This implies that the non-Hermitian part of $A_f(\lambda)$ has no contribution to the Wigner function as Pauli operators (apart from the identity) are traceless. However the non-Hermitian part of $A_f(\lambda)$ is important since when its operators compose into phase point operators for multiple qubits, they sometimes provide Hermitian operators that contribute to the Wigner function.
We know that $A_r(\lambda)$ is defined as the sum of observables $T(\lambda),$ where $\lambda=(\bold{q},\bold{p})$ such that $\bold{q}\cdot\bold{p}=0 \; \text{mod}\; 2.$  We can now see that also $\mathcal{H}(A_f(\lambda))$ is given by the sum of observables subjected to the same condition of having zero inner product between the components. This condition indeed singles out all the rebit observables, which are the only ones we are interested in.
Given an observable $T(\lambda)=Z(\bold{p})X(\bold{q})$ in $A_f(\lambda),$ with $\lambda=(\bold{q},\bold{p})$ and $\bold{q},\bold{p}\in\mathbb{Z}^n_d,$ it is Hermitian if and only if $T(\lambda)=T(\lambda)^{\dagger}.$ This means that \[T(\lambda)^{\dagger}=X(\bold{q})Z(\bold{p})=(-1)^{\bold{q}\cdot \bold{p}}T(\lambda),\] which holds if and only if $\bold{q}\cdot\bold{p}=0 \; \text{mod}\; 2.$ 

\end{proof}

In conclusion, by using one of the above Wigner functions, \eqref{DelfosseWf} or \eqref{factorisablePPO}, given the duality between states and measurement elements, the covariance and lemma \ref{WignerFunctionTransformation}, we can conclude that CSS rebit subtheory is Spekkens-representable. 
Moreover the definition of CSS-preserving transformations guarantees the closure property and the discrete Hudson's theorem for rebits (\emph{i.e.} non-negativity of the Wigner function of a state if and only if it is a CSS state \cite{Delfosse}) guarantees that it is maximal. Therefore the CSS rebit subtheory of quantum mechanics is a maximal Spekkens' subtheory.

\section{Spekkens' subtheories as toolboxes for state-injection}
\label{three}

We now prove that qubit stabilizer quantum mechanics can be obtained from a Spekkens' subtheory, in the sense that within Spekkens' subtheories it is possible to build a state-injection scheme that injects all the objects of qubit stabilizer quantum mechanics that are not in the subtheories. 
We need to understand which objects  we actually need to inject to reach the full qubit stabilizer quantum mechanics from a Spekkens' subtheory. 
Let us start by stating the list of instructions to construct the \textit{maximal} Spekkens' subtheory that corresponds to a given choice of $\gamma$ (in analogy with the framework of \cite{Rauss2}).\footnote{Notice that with the following construction a given $\gamma$ provides the \emph{maximal} Spekkens' subtheory, but the Wigner function from the same $\gamma$ can, obviously, be used to represent any smaller subtheory of the maximal Spekkens' subtheory.}
We recall that the function $\gamma$ uniquely specifies the function $\beta(\lambda,\lambda'),$ defined such that $T(\lambda)T(\lambda')=w^{\beta(\lambda,\lambda')}T(\lambda+\lambda').$ It results that the observable $T(\lambda)$ preserves positivity iff $\beta(\lambda,\lambda')=0\; \forall \; \lambda'\; s.t.\;  [\lambda,\lambda']=0,$ as proven in \cite{Rauss2}.

\begin{enumerate}
\item The function $\gamma$ uniquely defines the set of allowed observables, 
$\mathcal{M}=\{T(\bold\lambda)\;|\;\beta(\lambda,\lambda')=0\; \forall \; \lambda'\; s.t.\;  [\lambda,\lambda']=0\}.$ 
\item The set of allowed states $\mathcal{S}$ is given by the states corresponding to common eigenstates of $d^n$ commuting observables in $\mathcal{M}.$ 
\item The set of allowed gates is $\{U |\; U\rho U^{\dagger}=\rho'\in\mathcal{S} \; \forall \; \rho\in\mathcal{S}, \; and\; W_{U\rho U^{\dagger}}(\lambda)=W_{\rho}(S\boldsymbol{\lambda}+\bold{a})\; \forall \; \lambda\in\Omega\}.$ This is a subset of the Clifford unitaries.
\end{enumerate}

Let us point out that the above construction differs from \cite{Rauss2} in that it does not require tomographic completeness and it does require covariance of the Wigner functions. With respect to the Wallman-Bartlett 8-state model \cite{WallmanBartlet} the difference holds for analogous reasons. The 8-state model consists of a measurement non-contextual ontological model 
for \emph{one} qubit stabilizer quantum mechanics, where the one-qubit quantum states (and measurement elements) are represented as uniform probability distributions over an ontic space of dimension $8$ and the Clifford transformations (generated by the Hadamard $H$ and phase gate $S$) are represented by permutations over the ontic space. In the definition of the one qubit stabilizer distributions both the possible phase-point operators for a Wigner function are considered, \emph{i.e.} both the one with an even number of minuses $A_{+}(0,0)=1 + X + Y + Z,$ and the one with an odd number $A_{-}(0,0)=1 + X + Y - Z.$ See \cite{Galvao} for a more extensive description of these two possible pairs of single qubit Wigner functions.  

In addition to this, the proposed and straightforward generalisation of the 8-state model to more than one qubit, consists of considering the distributions built from the tensor products of the phase-point operators of the single qubit \cite{WallmanBartlet}. The resulting subtheory of qubit stabilizer quantum mechanics described by the model then becomes the one composed by all the product states of tensors of Pauli $X,Y,Z$ observables and all the \emph{local} Clifford unitary gates (generated by local $H$ and $S$). No entanglement is present. Nevertheless it is possible to reach universal quantum computation from this subtheory by performing measurement-based quantum computation \cite{MBQC} with a particular entangled cluster state \cite{Vega}\cite{Rauss2}. 

We have seen that the presence of non-covariant Clifford gates in the framework of Raussendorf et al. and the 8-state model is the main difference with respect to Spekkens' subtheories. Furthermore, the implementation of all the non-covariant Clifford unitaries would boost Spekkens' subtheories, composed by covariant Clifford gates, to the full qubit stabilizer quantum mechanics. 

\newtheorem{Main}{Theorem}
\begin{Main}\label{Main}
Qubit stabilizer quantum mechanics can be obtained from a Spekkens' subtheory via state-injection: all the possible non-covariant Clifford gates can be state-injected via a circuit made of objects in the Spekkens' subtheory.
\end{Main}

\begin{proof}

In order to prove that qubit stabilizer quantum mechanics can be obtained from a Spekkens' subtheory via state-injection we need to show that all the objects needed for injecting any non-covariant Clifford gate are present in at least one Spekkens' subtheory.  We recall that in order to generate the whole Clifford group we need, in addition to the $CNOT,$ also the generators of the local single gates, \emph{e.g.} the usual phase and Hadamard gates, $S,H.$
Let us consider the following subtheory, which corresponds to the CSS rebit subtheory, equations \eqref{CSSobservables} and \eqref{CSSgates}, with no global Hadamard gates:
\begin{itemize}
\item The allowed observables are, analogously to equation \eqref{CSSobservables}, non-mixing tensors of $X$ and $Z$ Pauli operators, $\mathcal{M} = \{X(\bold{q}),Z(\bold{p})|\bold{q},\bold{p} \in \mathbb{Z}_2^n\}.$ 
\item The allowed gates are the ones generated by the $CNOT$ and the Pauli rotations $X,Z,$ \emph{i.e.} \begin{equation}\label{CSSnoHgates} \mathcal{T}= \left\langle CNOT(i,j),X_i,Z_i\right\rangle.\end{equation}
\item The allowed states are, as usual, the eigenstates of the allowed observables. 
\end{itemize}
This is a smaller subtheory than CSS rebit (the difference being the absence of the global Hadamard gate). It possesses all the objects needed for state-injection of non-covariant gates. The $Z$ observables and the $CNOT$ gate are present. The correction gates are always Pauli gates, as for any injected Clifford unitary $U,$ even when  $U$ is non-covariant, $UX^{\otimes n}U^{\dagger},$ by definition of a Clifford gate, gives back a Pauli gate. All the objects of this subtheory can be non-negatively represented by the Wigner functions for the CSS rebit theory of the previous section and also in Spekkens' toy model, as shown in figure \ref{SpekSubToffoli}. Therefore this subtheory is closed and Spekkens representable, \emph{i.e.} a Spekkens' subtheory, and it is possible for it to reach universal quantum computation via state-injection, as proven in details in the next section.

Here, we will first show that we can obtain the whole Clifford group. Once we have it, we can map any of the allowed states and observables to any in qubit stabilizer quantum mechanics. The whole Clifford group can be achieved by first injecting the $CZ$ gate, as shown in figure \ref{SchemeCZ}, that also provides a construction for the Hadamard gate (figure \ref{HadamardToo}) and then the phase gate $S.$ The correction gate for the state-injection scheme of the $S$ gate is given by the Pauli $Y$ gate, which is present, up to a global phase, in our Spekkens' subtheory as a composition of $X$ and $Z$ Pauli rotations. Notice that, even if  state-injection, as defined in \ref{definition}, allows us to only inject diagonal unitary gates, we can  obtain all the non-diagonal gates because we have a full generating set of gates for Clifford unitaries (via the Hadamard gate). 
The other peculiarity of the Spekkens' subtheory we are considering is that it is minimal, as we cannot remove any object from it without denying the possibility of achieving universal quantum computation. 
It contains only the elements needed for the state injection scheme as defined in \ref{definition}. The other schemes that are not minimal, like in \cite{Howard,Delfosse}, have more components than strictly needed for the state-injection scheme defined in \ref{definition}.
To see that this scheme is minimal, note that if one were to remove the $X(\bold{q})$ observables, it would not be possible to obtain the Hadamard gate. Moreover Spekkens' subtheories where the observables are tensors of a single Pauli operator and gates that preserve the eigenbasis of that operator, do not allow any state-injection scheme to inject objects outside of them. 

\begin{figure}
\centering
{\includegraphics[width=.49\textwidth,height=.21\textheight]{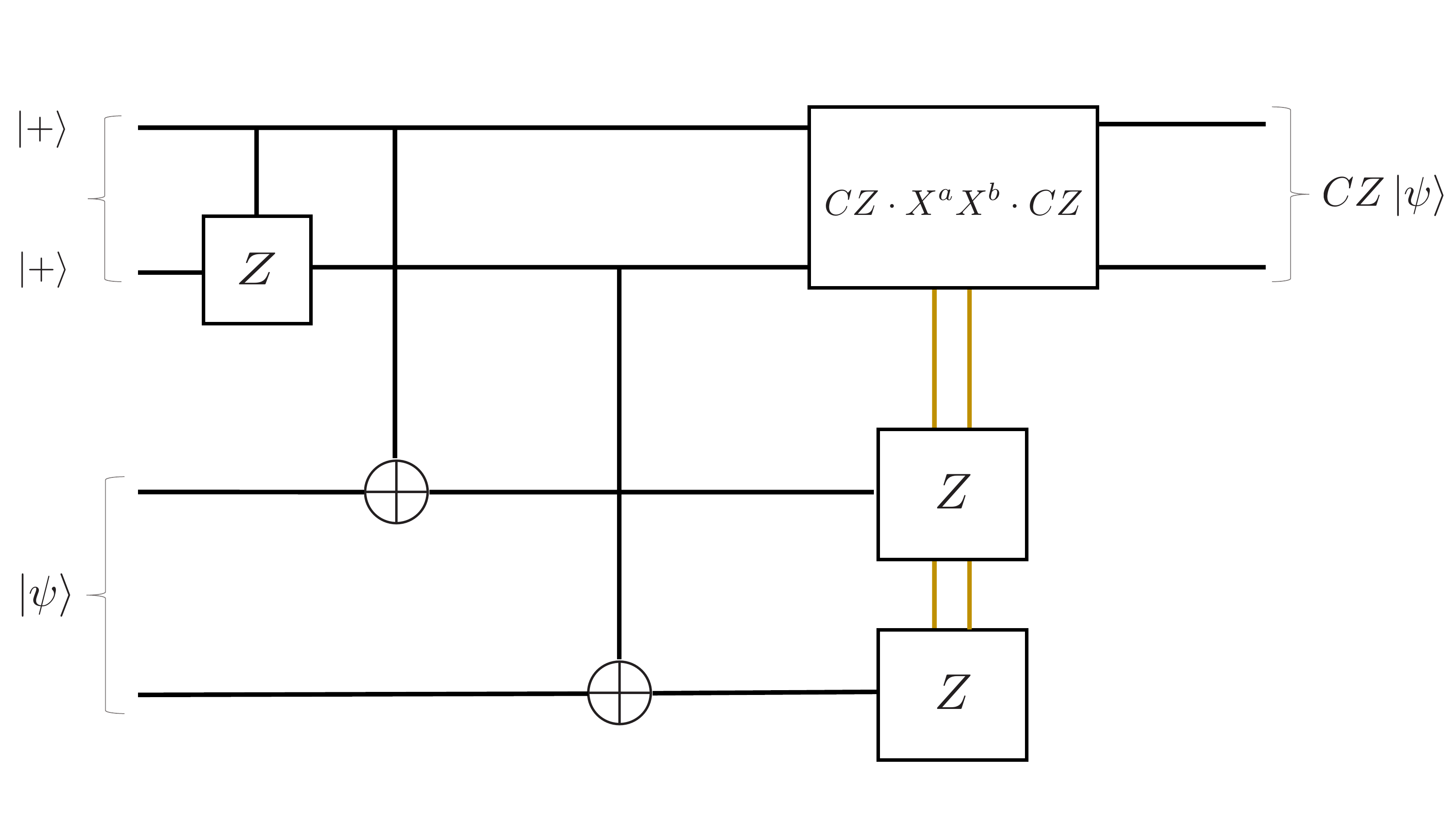}}
\caption[CZ injection.]{\footnotesize{\textbf{CZ injection.} The first injection of our scheme is the injection of the $CZ\left | ++\right \rangle$ state, needed in order to perform the correction in the injection scheme for the CCZ gate and to produce the Hadamard gate. In the figure above the correction is $CZ\cdot X^{a}X^{b}\cdot CZ$ conditioned on obtaining $x,y$ outcomes from the $Z$ measurements, where $(-1)^a=x$ and $(-1)^b=y$. For example if the outcomes are $x=1,y=-1,$ the correction is $CZ \cdot \mathbb{I}X\cdot CZ=XZ,$ which is an allowed gate in our subtheory.}}
\label{SchemeCZ}
\end{figure}
\begin{figure}
\centering
{\includegraphics[width=.4\textwidth,height=.165\textheight]{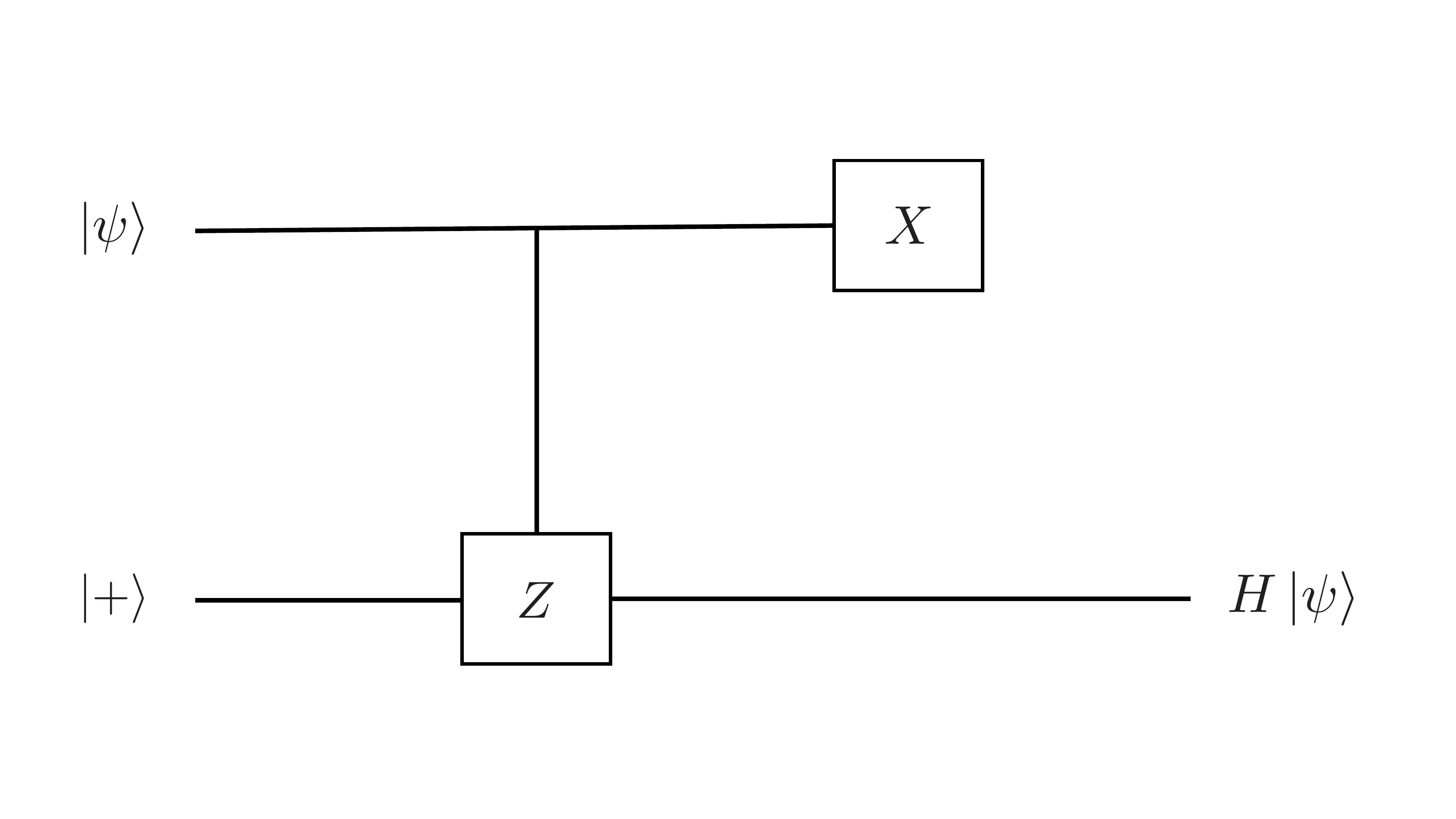}}
\caption[Hadamard gate via CZ.]{\footnotesize{\textbf{Hadamard gate via CZ.} The Hadamard gate can be obtained once the $CZ$ gate is available. Obtaining $H$ from $CZ$ and $X$ measurements was originally shown in \cite{BJS}.}}
\label{HadamardToo}
\end{figure}

\end{proof}

The above theorem guarantees that a state-injection scheme can always be recast in the following structure: $Spekkens' \; subtheory \; + \;Magic \; states\;\rightarrow\; UQC.$ Spekkens' model contains all the tools for state-injection of the non-covariant Clifford gates that appear in qubit stabilizer quantum mechanics. Note that this approach looks at the minimal non-contextual subtheories of stabilizer quantum mechanics where we can still reach universal quantum computation via state-injection. We see that each time these minimal subtheories are Spekkens' subtheories. In the next section we describe how the minimal subtheory of rebit stabilizer quantum mechanics used in the proof of theorem \ref{Main}, equations \eqref{CSSobservables} and \eqref{CSSnoHgates}, can reach universal quantum computation through the injection of CCZ magic states.

\section{State-injection schemes with CCZ states}
\label{four}

Reaching fault-tolerant universal quantum computation by exploiting Toffoli gates, or CCNOT - control control $X,$ goes back to Peter Shor in 1997 \cite{Shor}. It is known that the Toffoli gate (enough for universal classical computation) and Hadamard gate allow universal quantum computation \cite{Shi,Aharonov}, and, in a sense, this is the most natural universal set of gates, since the Toffoli enables all  classical computations, and just by adding the Hadamard gate, which generates superposition, we achieve universal quantum computation. The same result holds if we use the related  $CCZ,$ control control $Z,$ gate instead of the Toffoli gate. Other examples of fault tolerant universal quantum computation involving Toffoli state distillation have been proposed \cite{Eastin,Jones1,Jones2,Paetznick}. 
We here propose a scheme of $CCZ$ injection with the fewest possible objects in the free part of the computation. We reach that by also injecting the $CZ$ state before the CCZ.
The state-injection scheme is depicted in figure \ref{Scheme1}. 

\begin{figure}
\centering
{\includegraphics[width=.45\textwidth,height=.2\textheight]{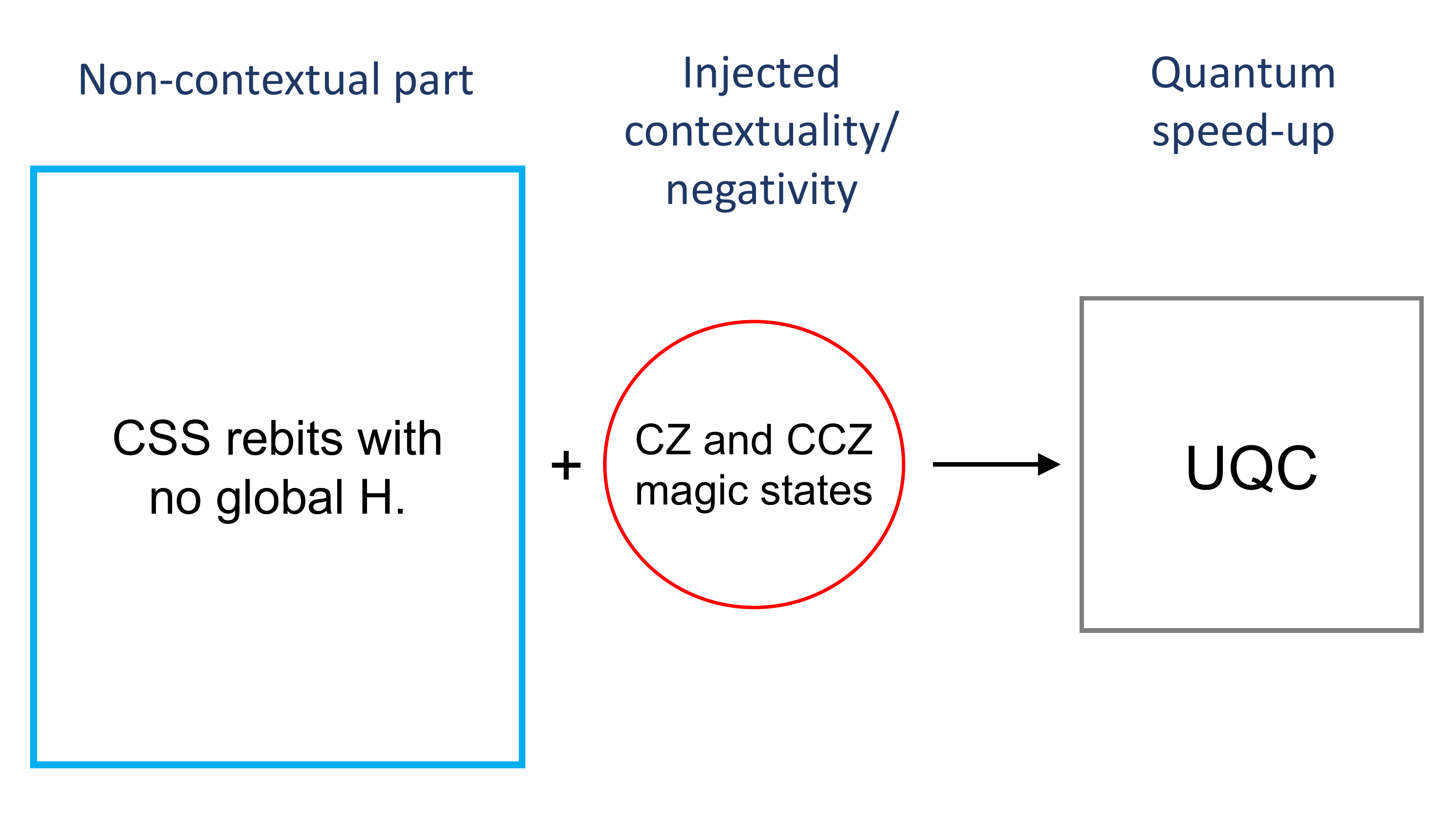}}
\caption[Novel state-injection scheme based on CCZ injection.]{\footnotesize{\textbf{Novel state-injection scheme based on CCZ injection.} By injecting first the $CZ$ state and then the $CCZ$ state we can boost the subtheory of rebit stabilizer quantum mechanics made of observables that are tensors of non-mixing Pauli $\mathbb{I},X,Z,$ and the gates generated by $CNOT$ and the Pauli rotations $X,Z$ to universal quantum computation.}}
\label{Scheme1}
\end{figure}

The free part is the one described in the proof of theorem \ref{Main} and defined by the equations \eqref{CSSobservables} and \eqref{CSSnoHgates}. 
This is the subtheory of CSS rebits with no global Hadamard, thus it is a Spekkens' subtheory. The state-injection scheme uses two state-injections: first the state $CZ\left | ++\right \rangle$ (figure \ref{SchemeCZ}), where the correction is given by the $CZ\cdot X^{a}X^{b}\cdot CZ$ conditioned on obtaining $x,y$ outcomes from the measurements of $Z$'s, where $(-1)^a=x$ and $(-1)^b=y$. Just to give an example, for outcomes $x=1,y=-1$ the correction is $CZ \cdot \mathbb{I}X\cdot CZ=XZ.$ 
Secondly, the injection of the state $CCZ\left | +++\right \rangle$ (figure \ref{SchemeCCZ}), where the correction is given by $CCZ\cdot X^{a}X^{b}X^{c}\cdot CCZ,$  with outcomes  $x,y,z$ of the measurments of $Z$'s such that $(-1)^a=x, \;(-1)^b=y$ and $(-1)^c=z,$ \emph{e.g.}  $CCZ\cdot X\mathbb{I}\mathbb{I}\cdot CCZ= X\cdot CZ$ if the outcomes are $-1,1,1.$ 
Notice that the injection of the $CZ$ also allows us to obtain the Hadamard gate, as shown in figure \ref{HadamardToo}. With Hadamard and CCZ gates we then have a universal set for quantum computation. The contextuality, which is not present in the subtheory of CSS rebits, is clearly present after the two injections that lead to universal quantum computation. 
\begin{figure}
\centering
{\includegraphics[width=.49\textwidth,height=.21\textheight]{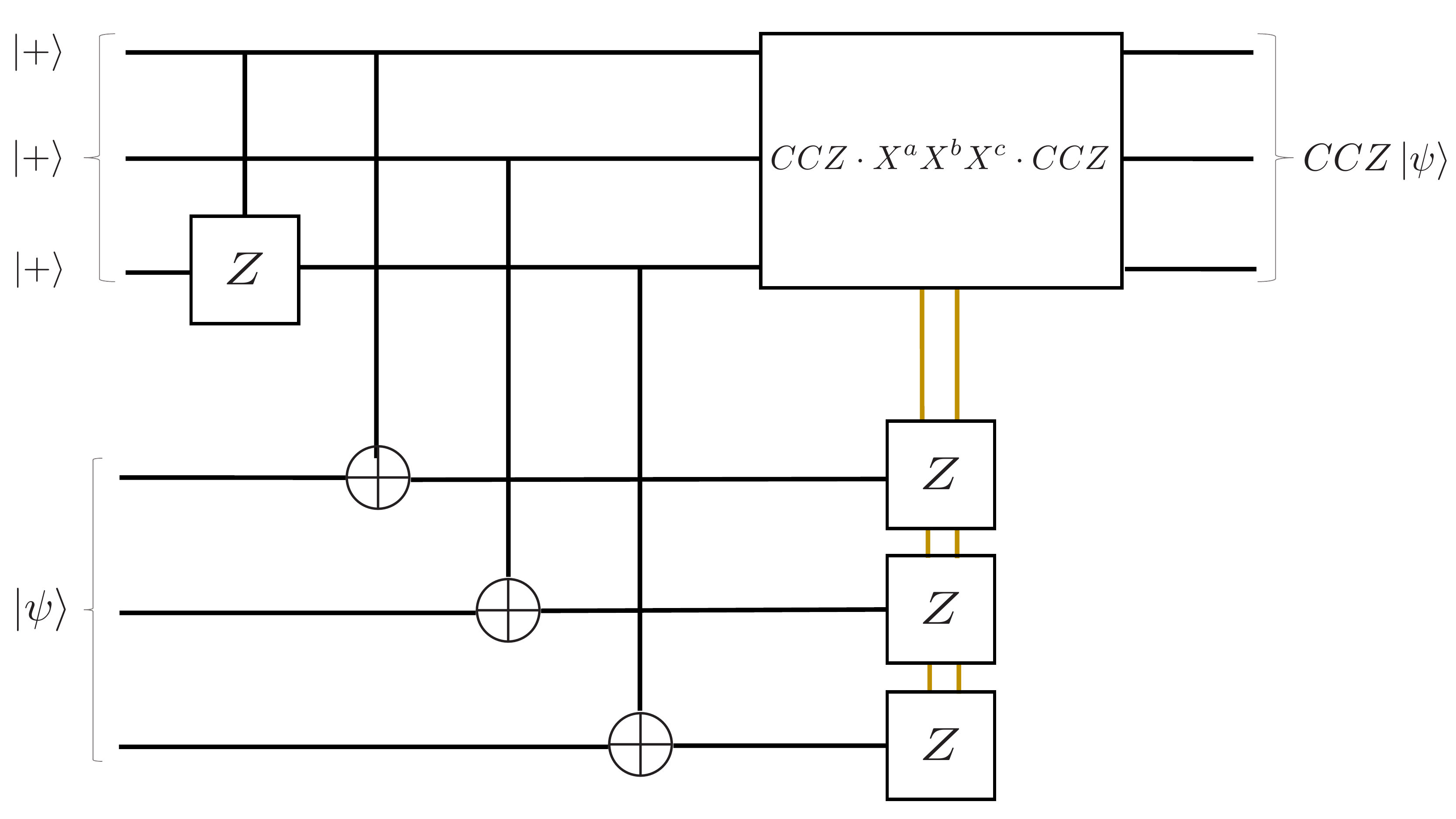}}
\caption[CCZ injection.]{\footnotesize{\textbf{CCZ injection.} The second injection of our scheme is the injection of the $CCZ\left |+++\right \rangle$ state. In the figure above the correction is $CCZ\cdot X^{a}X^{b}X^{c}\cdot CCZ$ conditioned on obtaining $x,y,z$ outcomes from the $Z$ measurements, where $(-1)^a=x,(-1)^b=y$ and $(-1)^c=z$. For example if the outcomes are $x=-1,y=1,z=1$ the correction is $CCZ\cdot X\mathbb{I}\mathbb{I}\cdot CCZ= X\cdot CZ,$ which is an allowed gate in our subtheory.}}
\label{SchemeCCZ}
\end{figure}

A few comments on the free part of the scheme are needed. As already said, it is \emph{minimal}, in the sense that it is not possible to remove any object from the free part of the computation without denying the possibility of obtaining universal quantum computation via state-injection. Also it is a strict subtheory of the CSS rebit, where we allow all the same objects apart from the global Hadamard. We argue that this is desirable, since in principle the Hadamard gate is a local gate; we want to keep only the entangling gates to have a global nature. 

\section{Proofs of Contextuality and state-injection}
\label{five}

The Spekkens' subtheory used for the CCZ state-injection scheme of the previous section allows us to establish a relation between the different resources injected and different proofs of contextuality. It is well known that within qubit stabilizer quantum mechanics we can obtain the Peres-Mermin square argument, which is a proof of state-independent contextuality, and the GHZ paradox, which is a proof of state-dependent contextuality \cite{Mermin,Peres,GHZ}. These proofs are not present within the Spekkens' subtheory, which, as we know, always witnesses the absence of any form of contextuality. We now explicitly show how the Peres-Mermin square and GHZ-paradox are obtained after the injection of either the $CZ$ gate or the $S$ gate. 
We also show that injection of the important non-Clifford $\frac{\pi}{8}$ gate $T$, in addition to the Peres-Mermin square and GHZ paradox (provided that we can apply the $T$ gate at least two times, as $T^2=S$), enables maximum violation of the CHSH inequality. These examples demonstrate that specific states injected to the minimal Spekkens' subtheory here considered, even if they do not provide universality, can be considered resources for specific and distinct manifestations of contextuality.

\begin{itemize}
\item \emph{Peres-Mermin square.} Let us consider the free Spekkens' subtheory of the CCZ injection scheme supplemented with the injection of the $CZ$ state. This allows us to construct a circuit to perform the Peres-Mermin square argument \cite{Mermin}. Let us first recall that the Peres-Mermin square (shown below) is one of the most intuitive and popular ways to illustrate the notion of Kochen-Specker contextuality.  

\begin{figure}[h!]
\centering
{\includegraphics[width=.33\textwidth,height=.14\textheight]{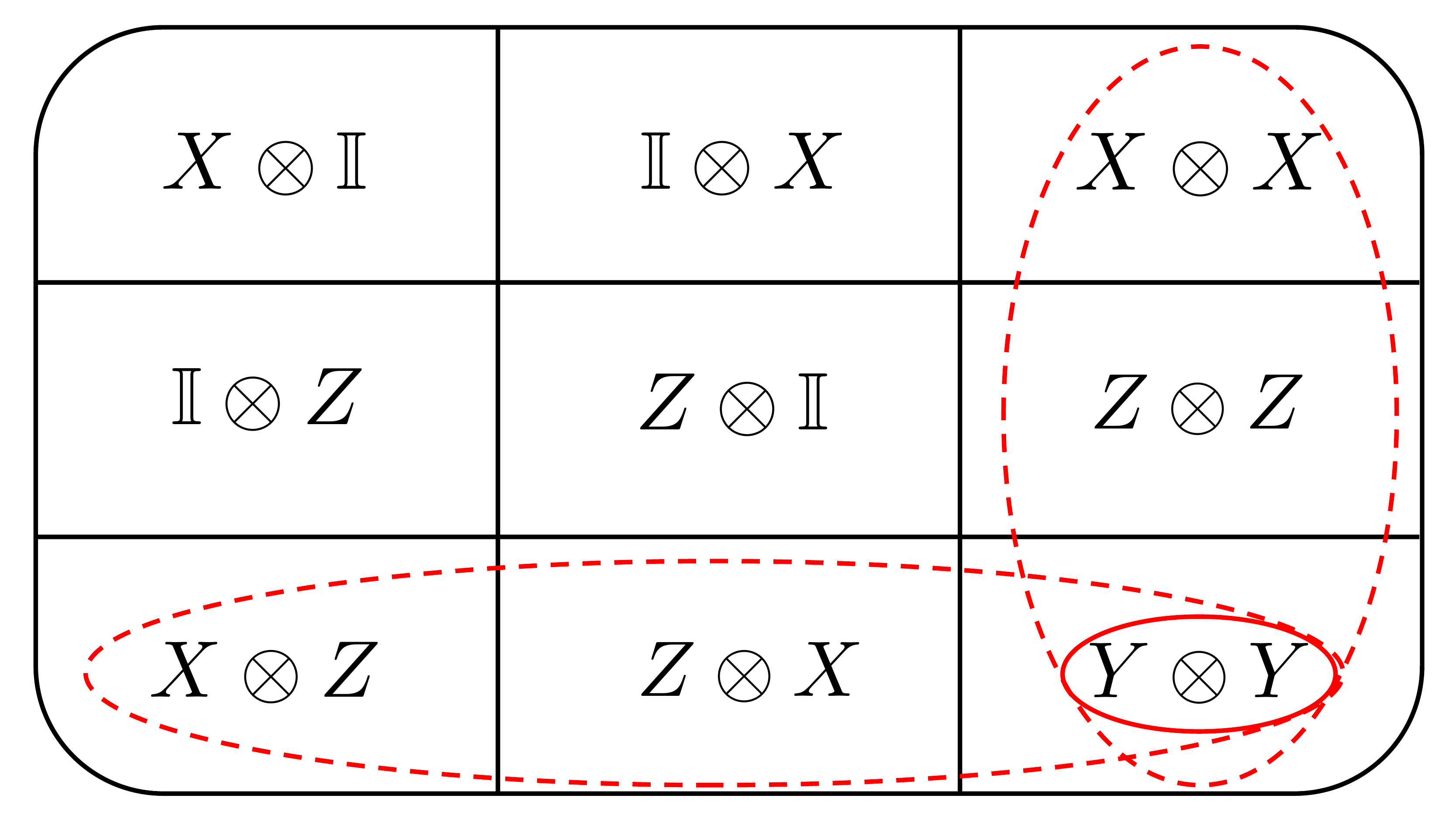}}
\label{MerminSquare}
\end{figure}

The square is composed by nine Pauli observables on a two-qubit system.\footnote{Notice that we will not write again the $"\otimes"$ symbol for the tensor product in order to soften the notation.} 
Each row and each column is composed by commuting (simultaneously measurable) observables. With the assumption that the functional relation between commuting observables is preserved in terms of their outcomes (\emph{e.g.} if an observable $C$ is the product of two observables $A,B$, also its outcome $c$ is the product of the the outcomes $a,b$ of $A,B$) and the outcome of each observable does not depend on which other commuting observables are performed with it (non-contextuality), the square shows that it is impossible to assign the outcome of each observable among all the rows and columns without falling into contradiction. For example, if we start by assigning values, say $\pm 1,$ to the observables starting from the first (top left) row on, the contradiction can be easily seen when we arrive at the last column and last row (red circles), that bring different results to the same observable $YY,$ as witnessed by the following simple calculation, $(XZ)\cdot (ZX) = YY,$ and $(XX)\cdot (ZZ) = -(YY).$ 
Kochen-Specker contextuality refers to the fact that the outcome of a measurement does depend on the other compatible measurements that we perform with it (\emph{i.e.} on the contexts). 

While in our original Spekkens' subtheory we are only allowed to perform the observables in the first two rows of the square, with the presence of the $CZ$ we can obtain the last row too, since $CZ\cdot X\mathbb{I} \cdot CZ= XZ,$ $CZ\cdot \mathbb{I}X \cdot CZ= ZX$ and $CZ\cdot XX \cdot CZ= YY.$  Figure \ref{CZMermin} shows a circuit where we can perform all the contexts of the Peres-Mermin square on an arbitrary input state $\left |\psi\right\rangle$ by just using objects belonging to the Spekkens' subtheory and CZ injections.

\begin{figure*}[]
\centering
{\includegraphics[width=.9\textwidth,height=.34\textheight]{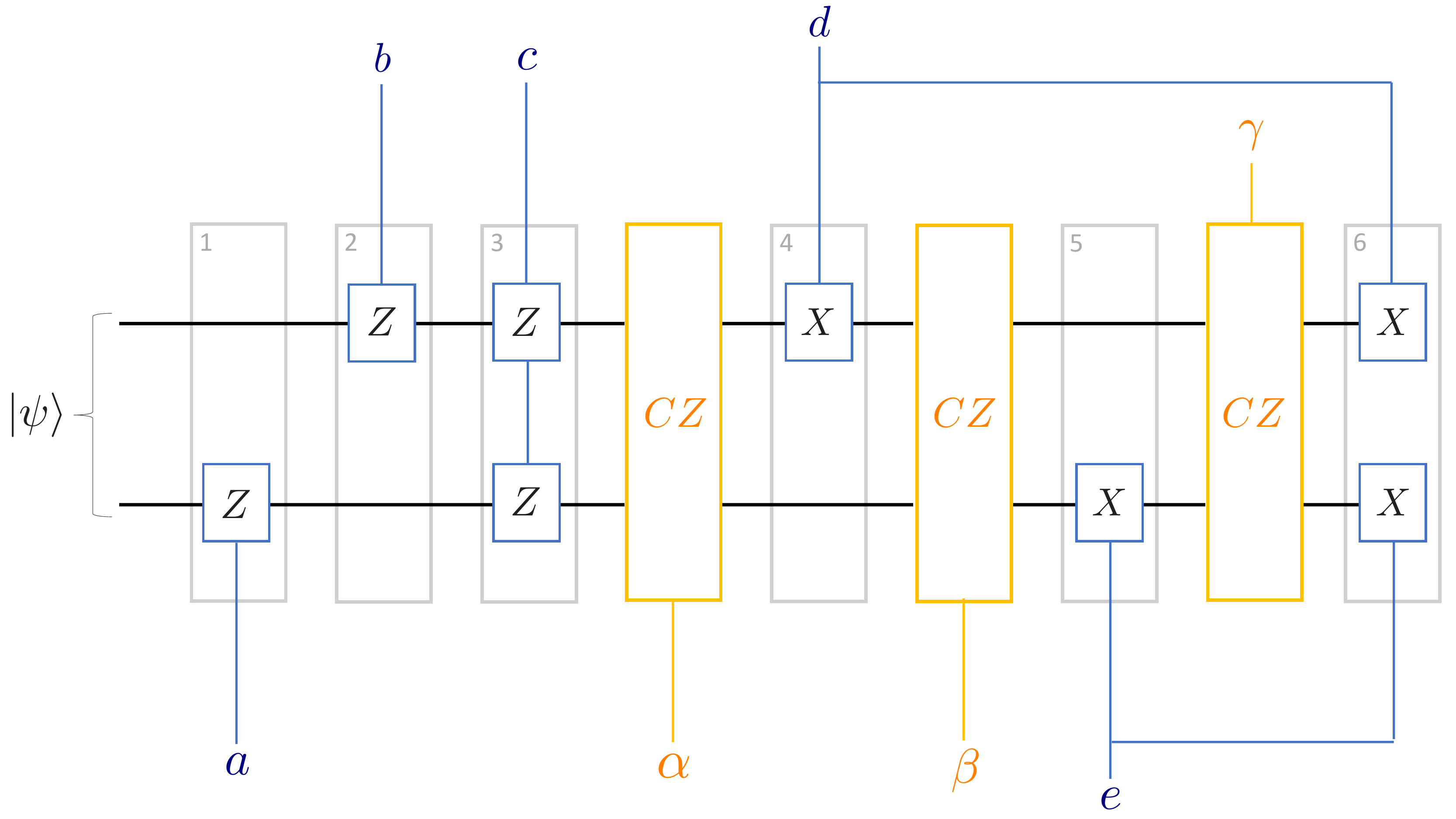}}

\caption[Peres-Mermin square via Spekkens' subtheories and CZ gates.]{\footnotesize{\textbf{Peres-Mermin square via Spekkens' subtheories and CZ gates.} The above circuit provides a way of implementing all the contexts of the Peres-Mermin square. Each block denoted by $CZ$ corresponds to the injection scheme of figure \ref{SchemeCZ} endowed also with a swap gate (which is present in our Spekkens' subtheory as it can be made of a series of three alternated CNOT gates) in order to set the output state $CZ\left |\psi\right\rangle$ as a precise modification of the input state $\left |\psi\right\rangle$ (and not of the ancillary resource state $CZ\left |++\right\rangle$). Each context can be selected according to some combinations of the classical control bits $a,b,c,d,e,\alpha,\beta,\gamma$ that can take values in $\{0,1\}.$ The value $0$ indicates that the corresponding gate is not performed, while the value $1$ that it is performed. At the end of every grey block (labelled by numbers) we assume that we can read the output outcome. The three row contexts of the Peres Mermin square are identified by the variables $(d,e),$ $(a,b,c)$ and $(\alpha,\beta,\gamma,d,e)$ assuming value one, respectively. The three column contexts are identified by $(a,d,\gamma)$ $(b,e,\gamma)$ and $(c,d,e,\gamma).$ Notice that in the last case where we implement the context $XX,ZZ,YY,$ the measurement of XX is implemented by performing $\mathbb{I}X$ first and then $X\mathbb{I},$ and in this case we consider the outputs related to the blocks labelled by $3,5,6.$}}

\label{CZMermin}
\end{figure*}

We can obtain the Peres-Mermin square argument also via the injection of the $S$ gate. This time the observables considered in the square are $\mathbb{I}X,X\mathbb{I}X,XX,Y\mathbb{I},\mathbb{I}Y,\mathbb{I}Y,YX,XY,ZZ.$ The ones containing $Y$ can be obtained by applying $S$ to the $X$ observable, while the others are already present in our Spekkens' subtheory.

\item \emph{GHZ paradox \cite{GHZ}.} In order to obtain the GHZ paradox we need to be able to create the GHZ state $\frac{\left |000  \right\rangle+ \left |111 \right\rangle}{\sqrt{2}},$ already present in our Spekkens' subtheory, and measure the mutually commuting observables $XXX,XYY,YXY,YYX.$ The GHZ state is the common eigenstate of these four operators, with the eigenvalues being $+1, -1, -1, -1$ respectively. While the first observable $XXX$ is already present in our Spekkens' subtheory, the others can be obtained either by local $S$ gate or $CZ$ gate on two of the three single Pauli operators composing each observable. By considering these observables, the quantum predictions are in conflict with any non-contextual hidden variable model that assigns definite pre-existing values, $+1$ and $-1,$ to the local Pauli observables $X,Y.$  Let us denote these definite values as $\lambda_{x1},\lambda_{x2},\lambda_{x3},\lambda_{y1},\lambda_{y2},\lambda_{y3}$ in correspondence of each local Pauli $X$ and $Y.$ The product of the three observables $XYY,YXY,YYX,$ that must yield the outcome $-1,$ in the hidden variable model means the following expression $\lambda_{x1}\lambda_{x2}\lambda_{x3}\lambda_{y1}^2\lambda_{y2}^2\lambda_{y3}^2=-1.$ However this is in neat contradiction with the outcome of $XXX$ which is $+1,$ and corresponds to $\lambda_{x1}\lambda_{x2}\lambda_{x3}=+1.$

\item \emph{CHSH argument \cite{CHSH}.} If we consider our Spekkens' subtheory with the addition of the $T$ gate we can obtain the maximum violation of the CHSH inequality. In the CHSH game a referee asks questions $x,y\in\{0,1\}$ to Alice and Bob respectively, who agree on a strategy beforehand to then answer $a,b\in\{0,1\}$ respectively. They win the game if $xy=a\oplus b,$ where the sum is meant to be modulo 2. The best classical strategy for them consists of always answering $a=b=0,$ which means winning the game with a probability of $75\%.$ By exploiting quantum states and measurements they can do better than that. It results that if they share the Bell state  $\frac{\left |00 \right\rangle- \left |11 \right\rangle}{\sqrt{2}},$ which is a $-1$ eigenstate of $XX$ and $+1$ eigenstate of $ZZ,$ and they perform the appropriate observables $A_q,B_q$ (depending on which question $q$, $0$ or $1,$ the referee asks them) they can win the game with the maximum quantum probability of about $85\%.$ Notice that the Bell state that we consider is present in our Spekkens' subtheory. 
The observables, which are provided by the presence of the $T$ gate, are $A_0=Y, \; A_1=X, \; B_0=TYT^{\dagger}=\frac{Y-X}{\sqrt{2}},\; B_1=TXT^{\dagger}=\frac{X+Y}{\sqrt{2}}.$ 
The probability that Alice and Bob win minus the probability that they lose is $\frac{1}{4}(\left\langle A_0B_0 + A_0B_1 + A_1B_0 - A_1B_1\right\rangle)=\frac{1}{\sqrt{2}},$ as $\left\langle A_0B_0\right\rangle=\left\langle A_0B_1\right\rangle=\left\langle A_1B_0\right\rangle=-\left\langle A_1B_1\right\rangle=\frac{1}{\sqrt{2}}.$ Therefore the probability of winning is $\frac{1}{2}+\frac{1}{2\sqrt{2}}\approx0.85.$
\end{itemize}


\section{Conclusions}

In this work we studied the subtheories of Spekkens' toy theory that are compatible with quantum mechanics, in the sense of making the same operational prediction. We identified these as the closed subtheories within Spekkens' toy theory that have non-negative and covariant Wigner function representations. Stabilizer quantum mechanics is the maximal Spekkens' subtheory in odd dimensions, as it corresponds to the full Spekkens' theory. This is not true for qubits, as stabilizer quantum mechanics is contextual and the toy theory does not reproduce its statistics. 
We used Spekkens' subtheories as a framework to describe the known examples of state-injection schemes for quantum computation on qudits of odd prime dimensions \cite{Howard} and rebits \cite{Delfosse} with contextuality as an injected resource, in the sense that they fit into the scheme of figure \ref{Ours}, \emph{i.e.} $ Spekkens' \; subtheory\; + \;Magic \; states\;\rightarrow\; UQC,$ where Spekkens' subtheories embody the non-contextual part of the computational model and the contextuality arises in the injection of the magic states. 
Furthermore, we showed that Spekkens' subtheories plus state-injection allow us to achieve all state-injection schemes where free operations are from the Clifford group (the standard state-injection schemes studied in the fault-tolerance literature). Theorem \ref{Main} proves that all of qubit stabilizer quantum mechanics can be obtained from a Spekkens' subtheory via state-injection, and  thus other state-injection schemes of quantum computation with Clifford group gates as the free operations can be mapped to our framework via a sequence of injections. 
In order to prove theorem \ref{Main} we constructed a Spekkens' subtheory which is a strict subtheory of the CSS rebit theory (used in \cite{Delfosse}) and provided a novel state-injection scheme to reach universal quantum computation by injections of $CZ$ and $CCZ$ states. This subtheory is minimal, meaning that it is not possible to remove any object from it without denying the possibility of reaching universal quantum computation via state-injection. 
By analysing the different injection processes in the above scheme we also associated different proofs of contextuality to specific state-injections of non-covariant gates. In particular we explicitly showed how the $CZ$ and $S$ gates are resources for the Peres-Mermin proof of contextuality and the GHZ paradox, and how the $T$ gate, used in the most popular state-injection schemes \cite{BravijKitaev}, is a resource  for maximal violations of the CHSH inequality. 

With respect to previous related works, we often referred to Raussendorf et al's framework \cite{Rauss2}, since this is very general and includes, for example, also the subtheory of qubit stabilizer quantum mechanics that arises from the 8-state model  of Wallman and Bartlett \cite{WallmanBartlet}. Raussendorf et al's framework differs from ours in that it does require tomographic completeness and it does not demand covariant Wigner functions.
In our framework, we  preferred to use the tools provided by Spekkens' toy theory, which has a less abstract structure, being an intuitive and fully non-contextual hidden variable model. 
The state-injection schemes we considered are the ones developed in \cite{State-Injection}, which means that we are not considering more general ways to use states as a resource such as the cluster state computation, considered in \cite{Rauss2}. An open question is how to extend theorem \ref{Main} to these more general schemes. A suggestion in this direction comes from the example of cluster state computation provided in \cite{Vega} and \cite{Rauss2}. It consists of a non-contextual free subtheory made of tensors of $X,Y,Z$ Pauli observables, their product eigenstates and all the local Clifford gates, and the resource is a specific entangled cluster state. The free subtheory in this case is not a Spekkens' subtheory, as the $S$ and $H$ gates are not covariant if we allow all the product eigenstates of tensors of $X,Y,Z$ Pauli observables. However if we remove these local gates we can still implement the computational scheme (which never needs to use those gates) and obtain universal quantum computation with the same resource state. In the latter case the free part is now a Spekkens' subtheory. Therefore this example can be actually recast in our framework of $Spekkens' \; subtheory\;+\; Magic\;state\rightarrow UQC.$
Finally we point out that all previous works \cite{Vega,Rauss2,WallmanBartlet} look at the biggest non-contextual subtheories of stabilizer quantum mechanics that have state-injection schemes with contextuality as a resource. Here instead we focus also on the smallest free subtheories from which it is still possible to reach universal quantum computation via state-injection.

We believe that the results presented here suggest some future projects. The central role of covariance in our work suggests that its relationship with non-contextuality deserves further study. A recent work on contextuality in the cohomological framework could provide the right tools to address this question \cite{Okay}. In particular, covariance seems to be strictly related to Spekkens' transformation non-contextuality \cite{Spek3}. As an example of this, the single qubit stabilizer quantum mechanics, already argued to be not covariant, shows transformation contextuality (even if a preparation and measurement non-contextual model for it - \emph{e.g.} the 8-state model - exists) \cite{Piers}.
One significant question which remains open is to clarify  which notion of contextuality is the best one to capture resources for universal quantum computation as, for example, it is known that qubit stabilizer quantum mechanics, despite being contextual, is efficiently classically simulatable \cite{GottesKnill}. It would be desirable to match the notion of non-classicality in quantum foundations, namely contextuality, with the notion of non-classicality in quantum computation, \emph{e.g.} non-efficient classical simulatability. 
Finally we think that it would be interesting to extend Spekkens' toy theory to obtain a psi-epistemic ontological model \cite{Harrigan} of $n-$qubit stabilizer quantum mechanics. Possibly some extensions of the 8-state model can achieve this. If so, the  epistemic restrictions on which such models were built might be of independent interest. Such work would contribute to further understanding the subtle relationship between contextuality and the computational power of universal quantum computation.

\section{Acknowledgements}

We would like to thank Robert Raussendorf, Juan Bermejo-Vega, Piers Lillystone, Hammam Qassim, Nicolas Delfosse, Cihan Okay, Joel Wallman and Robert Spekkens for helpful discussions. We would also like to thank Nadish De Silva for useful considerations at the first stage of the project. This work was supported by EPSRC Centre for Doctoral Training in Delivering Quantum Technologies [EP/L015242/1].


\end{document}